\documentclass{article}

\usepackage{arxiv}
\usepackage{url}
\usepackage{fancyhdr}
\usepackage{multirow}
\usepackage{diagbox}
\usepackage{algorithm}
\usepackage[noend]{algpseudocode}
\usepackage{subfigure}
\usepackage[T1]{fontenc}

\RequirePackage{algorithm}
\usepackage[backref]{hyperref}

\makeatletter
\newenvironment{breakablealgorithm}
  {
    \begin{center}
      \refstepcounter{algorithm}
      \hrule height.8pt depth0pt \kern2pt
      \parskip 0pt
      \renewcommand{\caption}[2][\relax]{
        {\raggedright\textbf{\fname@algorithm~\thealgorithm} ##2\par}%
        \ifx\relax##1\relax 
          \addcontentsline{loa}{algorithm}{\protect\numberline{\thealgorithm}##2}%
        \else 
          \addcontentsline{loa}{algorithm}{\protect\numberline{\thealgorithm}##1}%
        \fi
        \kern2pt\hrule\kern2pt
     }
  }
  {
     \kern2pt\hrule\relax
   \end{center}
  }
\makeatother
\usepackage{xcolor}
\definecolor{lightgray}{gray}{0.893}
\usepackage{colortbl}

\usepackage{mathtools} 
\usepackage{booktabs} 
\usepackage{tikz} 
\usepackage{bm}
\usepackage{amssymb,amsmath,amsthm}

\newtheorem{definition}{Definition}
\newtheorem{theorem}{Theorem}
\newtheorem{corollary}{Corollary}
\newtheorem{assumption}{Assumption}

\title{Learning to Coordinate with Anyone}

\author{%
 Lei Yuan$^{1, 2}$, Lihe Li$^{1}$, Ziqian Zhang$^1$, Feng Chen$^1$, Tianyi Zhang$^{1}$, Cong Guan$^1$,
 {Yang Yu}$^{1,2,}$\thanks{Corresponding Author}, Zhi-Hua Zhou$^1$\\
  $^1$ National Key Laboratory for Novel Software Technology, Nanjing University\\
  $^2$ Polixir.ai\\
  \texttt{\{yuanl, lilh, zhangzq, chenf, zhangty, guanc\}@lamda.nju.edu.cn},\\ \texttt{yuy@nju.edu.cn}, \texttt{zhouzh@lamda.nju.edu.cn}
}

\begin{document}

\maketitle
\date{}
\begin{abstract}
In open multi-agent environments, the agents may encounter unexpected teammates. Classical multi-agent learning approaches train agents that can only coordinate with seen teammates. Recent studies attempted to generate diverse teammates to enhance the generalizable coordination ability, but were restricted by pre-defined teammates. In this work, our aim is to train agents with strong coordination ability by generating teammates that fully cover the teammate policy space, so that agents can coordinate with any teammates. Since the teammate policy space is too huge to be enumerated, we find only \emph{dissimilar} teammates that are \emph{incompatible} with controllable agents, which highly reduces the number of teammates that needed to be trained with. However, it is hard to determine the number of such incompatible teammates beforehand. We therefore introduce a continual multi-agent learning process, in which the agent learns to coordinate with different teammates until no more incompatible teammates can be found. The above idea is implemented in the proposed Macop (\textbf{M}ulti-\textbf{a}gent \textbf{co}mpatible \textbf{p}olicy learning) algorithm. We conduct experiments in 8 scenarios from 4 environments that have distinct coordination patterns. Experiments show that Macop generates training teammates with much lower compatibility than previous methods. As a result, in all scenarios Macop achieves the best overall coordination ability while never significantly worse than the baselines, showing strong generalization ability.
\end{abstract}

\section{Introduction}

Cooperative Multi-Agent Reinforcement Learning (MARL)~\cite{oroojlooy2023review} has garnered significant attention due to its demonstrated potential in various real-world applications. Recent studies have showcased MARL's exceptional performance in tasks such as pathfinding\cite{sartoretti2019primal}, active voltage control~\cite{DBLP:conf/nips/WangXGSG21}, and dynamic algorithm configuration~\cite{xue2022multiagent}. However, these achievements are typically made within closed environments where teammates are pre-defined. The system will suffer from coordination ability decline when deploying the trained policies in real-world scenarios, where agents may encounter unexpected teammates in such open environments~\cite{zhou2022open}. 

Training with diverse teammates presents a promising avenue for tackling the aforementioned challenge. Various methods have emerged in domains such as ad-hoc teamwork~\cite{mirsky2022survey}, zero-shot coordination \cite{treutlein2021new}, and few-shot teamwork \cite{fosong2022few}.
Addressing this challenge effectively involves two crucial factors. Firstly, to enhance generalization and avoid overfitting to specific partners, it is essential for agents to be exposed to diverse teammates during the training process. Diversity can be achieved through various techniques, such as hand-crafted policies \cite{papoudakis2021agent}, object regularizers designed among agents \cite{derek2021adaptable, lupu2021trajectory, lipo}, or population-based training (PBT) \cite{strouse2021collaborating, xue2022heterogeneous}. 
Secondly, when dealing with multiple teammates, especially in the context of multi-modal scenarios, specialized consideration is necessary. Naive approaches, like self-play (or "self-training")~\cite{tesauro1994td, silver2018general}, Fictitious Co-Play (FCP)~\cite{heinrich2015fictitious, strouse2021collaborating}, or coevolving agent and partner populations~\cite{xue2022heterogeneous}, have been explored (See related work in App.~\ref{relatedwork}). 
Nevertheless, complex scenarios often present substantial challenges arising from both the complexity and vastness of the teammate policy space. On one hand, enumerating all possible teammate groups is a daunting task, and training the agents can be time-consuming. On the other hand, even when we pre-define only representative and diverse teammates, we may still accidentally omit some instances. The exact number of such teammates cannot be determined in advance as well. This prompts a crucial question: Can we design a more efficient training paradigm that ensures our controllable agents are trained alongside partners in a policy space that guarantees coverage, ultimately enabling high generalization and effective coordination ability with diverse teammates? 

To tackle the mentioned issue, we propose a novel coordination paradigm known as Macop, with which we can obtain a multi-agent compatible policy via incompatible teammates evolution. The core principle of Macop is the adversarial generation of new teammate instances, which are strategically crafted to challenge and refine the ego-system's (the agents we control) coordination policy. However, the exact number of representative teammates can not be determined beforehand, and maintaining a sufficiently diverse population requires significant computing and storage resources. We therefore
introduce Continual Teammate Dec-POMDP (CT-Dec-POMDP), wherein the ego-system is trained with groups of teammates generated sequentially until convergence is reached. Our approach is rooted in two crucial factors: instance diversity and incompatibility between the newly generated teammates and the ego-system. During the training process, we iteratively refine teammate generation and optimize the ego-system until convergence is reached. This approach empowers the ego-system, leading to a coordination policy capable of seamlessly handling a wide array of team compositions and promptly adapting to new teammates.

We conduct experiments on different MARL benchmarks that have distinct coordination patterns, including Level-based Foraging (LBF)~\cite{lbf}, Predator-Prey (PP), Cooperative Navigation (CN) from MPE~\cite{maddpg}, and two customized maps from StarCraft Multi-agent Challenge (SMAC)~\cite{pymarl}. 
Experimental results show that our proposed Macop exhibits remarkable improvement in comparison to existing methods, achieving nearly 20\% average performance improvement in the conducted benchmarks compared to multiple baselines, and more experiments reveal it from multiple aspects.

\section{Problem Formulation} \label{problem formulation}

As we aim to solve a continual coordination problem, where the controllable agents are required to cooperate with diverse teammates which arise sequentially, we formalize it as a Continual Teammate Dec-POMDP (CT-Dec-POMDP) by extending the Dec-POMDP~\cite{oliehoek2016concise}. The CT-Dec-POMDP can be described as a tuple $\mathcal{M} = \langle\mathcal{N, S,A},P,\{\boldsymbol{\pi}_{\text{tm}}^k\}_{k=1}^\infty, m,$ $\Omega, O,R,\gamma \rangle$, here $\mathcal{N} = \{1, \dots, n\}$,  $\mathcal{S}$, $\mathcal{A}=\mathcal A^1\times...\times\mathcal A^n$ and $\Omega$ are the sets of corresponding agents, global state, joint action, observation.  $P$ is the transition function, $\{\boldsymbol{\pi}_{\text{tm}}^k\}_{k=1}^\infty$ represents the $k$ groups of teammates encountered sequentially during the training phase until time $t$,  $m$ is the number of controllable agents, and $\gamma \in [0, 1)$ represents the discounted factor. At each time step, agent $i$ receives the observation $o^i=O(s, i)$ and outputs the action $a^i\in\mathcal A^i$.

Concretely, when training to cooperate with a group of teammates $\pi_{\text{tm}}^k$, the agents do not have access to previous teammates groups $\pi_{\text{tm}}^{k'}, k'=1,...,k-1$. However, they are expected to remember how to cooperate with all previously encountered teammates groups. For simplicity, we denote a group of teammates as "teammate" when no ambiguity arises. The training phase of cooperating with teammate $\pi_{\text{tm}}^k$ can be described as $\mathcal{M}_k = \langle\mathcal{N, S,A},P,\boldsymbol{\pi}_{\text{tm}}^k, m, \Omega,$ $O,R,\gamma \rangle$. The controllable agents $\boldsymbol{\pi}_{\text{ego}}=\{\pi^1_{\text{ego}}, ...,\pi^m_{\text{ego}}\}\in\Pi_{\text{ego}}=\otimes_{i=1}^m \Pi_i$ and the teammate $\boldsymbol{\pi}_{\text{tm}}^k=\{\pi_{\text{tm}}^{k, m+1},...,\pi_{\text{tm}}^{k, n}\}\in\Pi_{\text{tm}}=\otimes_{i=m+1}^n \Pi_i$ formulate a new joint policy $\langle\boldsymbol{\pi}_{\text{ego}},\boldsymbol{\pi}_{\text{tm}}^k\rangle$. 
The joint action $\langle\boldsymbol{a}_{\text{ego}},\boldsymbol{a}_{\text{tm}}^k\rangle=\langle\boldsymbol{\pi}_{\text{ego}}(\boldsymbol{\tau}_{\text{ego}}),\boldsymbol{\pi}_{\text{tm}}^k(\boldsymbol{\tau}_{\text{tm}}^k)\rangle$ leads to the next state $s'\sim P(\cdot|s, \langle\boldsymbol{a}_{\text{ego}},\boldsymbol{a}_{\text{tm}}^k\rangle)$ and the global reward $R(s, \langle\boldsymbol{a}_{\text{ego}},\boldsymbol{a}_{\text{tm}}^k\rangle)$, where $\boldsymbol{\tau}_{\text{ego}}=\{\tau^i\}_{i=1}^m, \boldsymbol{\tau}_{\text{tm}}^k=\{\tau^i\}_{i=m+1}^n$. The controllable agents are optimized to maximize the expected return when cooperating with teammate $\boldsymbol{\pi}_{\text{tm}}^k$: 

\begin{equation}
    \begin{aligned}
 \max_{\boldsymbol{\pi}_{\text{ego}}} \mathcal{J}(\langle\boldsymbol{\pi}_{\text{ego}}, \boldsymbol{\pi}_{\text{tm}}^{k}\rangle) =  \mathbb E_{\boldsymbol{\tau}\sim\rho(\langle\boldsymbol{\pi}_{\text{ego}},\boldsymbol{\pi}_{\text{tm}}^{k}\rangle)}[G(\boldsymbol{\tau})],
    \end{aligned}
    \end{equation}
where $G(\boldsymbol{\tau})=\sum_{t=0}^T \gamma^t R(s_t, \boldsymbol{a}_t)$  is the  return of a joint trajectory. At the same time, for a formal characterization of the relationship between the policy space of  $\boldsymbol{\pi}_{\text{ego}}$ and $\boldsymbol{\pi}_{\text{tm}}$, we introduce the concept of complementary policy class:

\begin{definition}[complementary policy class]
For any sub policy $\boldsymbol{\pi}\in\Pi_{i:j}=\otimes_{h=i}^j\Pi_h, i\leq j$, we define its complementary policy class as $\Pi^{c}_{\boldsymbol{\pi}}=\otimes_{h=1}^{i-1}\Pi_h\times\otimes_{h=j+1}^n\Pi_h$. We denote the complementary policy class of controllable agents and the teammate as $\Pi^c_{\text{ego}}$ and $\Pi^c_{\text{tm}}$ for simplicity. We also refer $\mathcal{J}_{\text{sp}}(\boldsymbol{\pi}_{\text{ego}}) = \max_{\bar{\boldsymbol{\pi}}_{\text{tm}}\in\Pi_{\text{ego}}^c}\mathcal{J}(\langle\boldsymbol{\pi}_{\text{ego}}, \bar{\boldsymbol{\pi}}_{\text{tm}}\rangle)$ and $\mathcal{J}_{\text{sp}}(\boldsymbol{\pi}_{\text{tm}}) = \max_{\bar{\boldsymbol{\pi}}_{\text{ego}}\in\Pi_{\text{tm}}^c}\mathcal{J}(\langle\bar{\boldsymbol{\pi}}_{\text{ego}}, \boldsymbol{\pi}_{\text{tm}}\rangle)$ as ``self-play return'' of $\boldsymbol{\pi}_{\text{ego}}$ and $\boldsymbol{\pi}_{\text{tm}}$, respectively.
\label{def1}
\end{definition}


\section{Method}

In this section, we will present the detailed design of our proposed
method Macop (c.f. Fig.~\ref{macopflow}). First, we introduce a novel continual teammate generation module by combining population-based training
and incompatible policy learning (Fig.~\ref{macopflow}(a)). Next, we outline
the design of our continual coordination policy learning paradigm,
which consists of a shared backbone and a dynamic head expansion
module (Fig.~\ref{macopflow}(b)). These two phases proceed alternatively to train a
robust multi-agent coordination policy that is capable of effectively
cooperating with diverse teammates (Fig.~\ref{macopflow}(c)). 

\begin{figure}[H]
  \centering
    \includegraphics[width=0.66\textwidth]{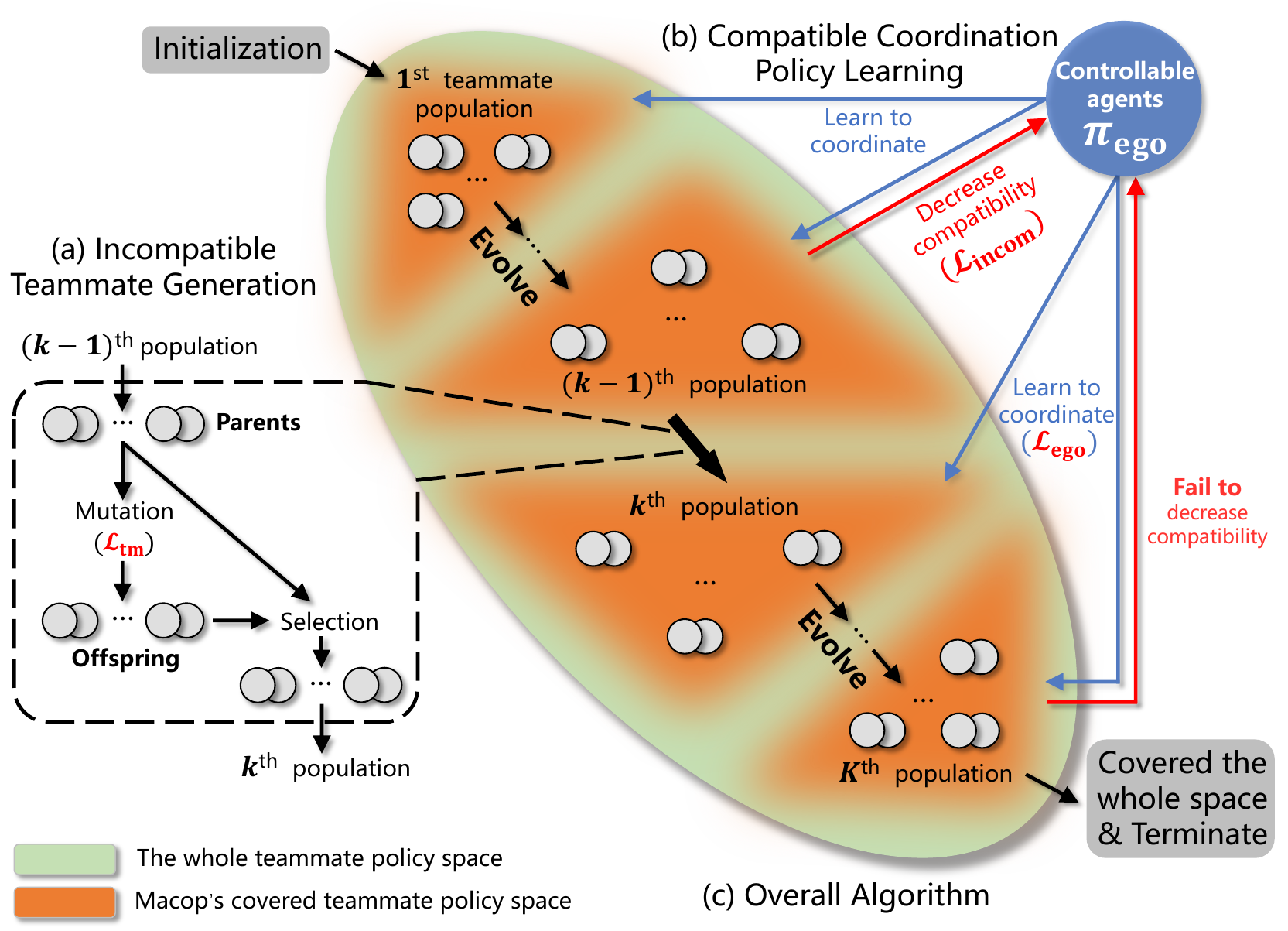}
  \caption{The overall workflow of Macop.}
      \label{macopflow}
\end{figure}

\subsection{Incompatible teammate generation}

The objective of Macop is to develop a joint policy that can effectively cooperate with diverse teammates. Since the policy space of teammate groups is too huge to be enumerated, we focus on identifying dissimilar teammate groups. 
To achieve this, we begin by establishing a complementary-policy-agnostic measure capable of effectively quantifying the similarity between two teammate groups, ensuring that it remains unaffected by complementary policies. In particular, we pair two teammate groups with any arbitrary complementary policy, as defined in Def.~\ref{def1}. These groups are considered similar if the probability of the trajectory produced by both groups surpasses a predefined threshold.
\begin{definition}[$\epsilon$-similar policies]
We measure the similarity between two different teammates $\boldsymbol{\pi}_{\text{tm}}^{i}, \boldsymbol{\pi}_{\text{tm}}^{j}$ with the probability of the trajectory induced by them when paired with any complementary policies. Specifically, for any fixed complementary policy $\bar{\boldsymbol{\pi}}\in\Pi_{\text{tm}}^c$, the probability of the trajectory produced by the joint policy $P(\boldsymbol{\tau}|\langle\bar{\boldsymbol{\pi}}, \boldsymbol{\pi}_{\text{tm}}\rangle)=\prod_{t=0}^{T-1}\bar{\boldsymbol{\pi}}(\bar{\boldsymbol{a}}_t|\bar{\boldsymbol{\tau}}_t)\boldsymbol{\pi}_{\text{tm}}(\boldsymbol{a}_{tm, t}|\boldsymbol{\tau}_{tm, t})P( $ $s_{t+1}|s_t, \langle\bar{\boldsymbol{a}}_t, \boldsymbol{a}_{tm, t}\rangle)$. Accordingly, we define the dissimilarity between the two teammates $d(\boldsymbol{\pi}_{\text{tm}}^i, \boldsymbol{\pi}_{\text{tm}}^j)=\max_{\boldsymbol{\tau}}|1-\frac{P(\boldsymbol{\tau}|\langle\bar{\boldsymbol{\pi}}, \boldsymbol{\pi}_{\text{tm}}^i\rangle)}{P(\boldsymbol{\tau}|\langle\bar{\boldsymbol{\pi}}, \boldsymbol{\pi}_{\text{tm}}^j\rangle)}|=\max_{\boldsymbol{\tau}}|1-\prod_{t=0}^{T-1}\frac{\boldsymbol{\pi}_{\text{tm}}^i(\boldsymbol{a}_{tm, t}|\boldsymbol{\tau}_{tm, t})}{\boldsymbol{\pi}_{\text{tm}}^j(\boldsymbol{a}_{tm, t}|\boldsymbol{\tau}_{tm, t})}|$. Teammates $\boldsymbol{\pi}_{\text{tm}}^{i}$ and $ \boldsymbol{\pi}_{\text{tm}}^{j}$ are $\epsilon-$similar policies if and only if $d(\boldsymbol{\pi}_{\text{tm}}^i, \boldsymbol{\pi}_{\text{tm}}^j)\leq \epsilon, 0\leq \epsilon\leq 1$, which implies that $1-\epsilon\leq\frac{P(\boldsymbol{\tau}|\langle\bar{\boldsymbol{\pi}}, \boldsymbol{\pi}_{\text{tm}}^i\rangle)}{P(\boldsymbol{\tau}|\langle\bar{\boldsymbol{\pi}}, \boldsymbol{\pi}_{\text{tm}}^j\rangle)}\leq 1+\epsilon, \forall\boldsymbol{\tau}$. 
\label{def2}
\end{definition}

Based on the Def.~\ref{def2} above, our approach involves the identification of representative teammate groups, ensuring that the dissimilarity between them surpasses the specified threshold $\epsilon$.  We continually generate such dissimilar 
teammate groups in order to gradually cover the space of teammate policies. Drawing inspiration from the proven efficacy of population-based training (PBT)~\cite{jaderberg2019human} and evolutionary algorithms (EA)~\cite{zhou2019evolutionary}, we adopt an evolutionary process to formulate the teammate generation process by maintaining a population of teammates $\mathcal{P}_{\text{tm}}=\{\boldsymbol{\pi}_{\text{tm}}^j\}_{j=1}^{n_p}$ under the changing controllable agents $\boldsymbol{\pi}_{\text{ego}}$. By ensuring that the teammate groups exhibits dissimilarity between instances in not only the current population but also previous ones, our aim is to systematically explore and cover the entire teammate policy space over time.

Specifically, in each generation, the current population is first initialized through a customized parent selection mechanism (details
provided later). We focus on promoting diversity within the teammate population, striving to enhance the dissimilarity between each individual, i.e., $\max \sum_{i\neq j } d(\boldsymbol{\pi}_{\text{tm}}^i, \boldsymbol{\pi}_{\text{tm}}^j)$. To achieve the goal mentioned, we take Jensen-Shannon divergence (JSD)~\cite{jsd} as a reliable proxy to effectively measure the dissimilarity between teammates' policies as is introduced in~\cite{romance}:
\begin{equation}
    \begin{aligned}
        \mathcal{L}_{\text{div}}=&\mathbb{E}_{s}[\text{JSD}(\{\boldsymbol{\pi}_{\text{tm}}^{i}\}_{i=1}^{n_p})]\\
        =&\mathbb{E}_{s}[\frac{1}{n_p}\sum_{i=1}^{n_p}D_{KL}(\boldsymbol{\pi}_{\text{tm}}^i(\cdot|s)||\bar{\boldsymbol{\pi}}_{\text{tm}}(\cdot|s))],
    \end{aligned}
    \label{loss_div}
\end{equation}
where $\bar{\boldsymbol{\pi}}_{\text{tm}}(\cdot|s)=\frac{1}{n_p}\sum_{i=1}^{n_p}\boldsymbol{\pi}_{\text{tm}}^i(\cdot|s)$ is the average policy of the population, and $D_{KL}$ is the Kullback-Leibler (KL) divergence between two distribution. We provide proofs that the JSD proxy is a certifiable lower bound of the original dissimilarity objective in App.~\ref{ProofsforTheorems}.

The advantages of JSD are immediately apparent. Unlike TV divergence or KL divergence, which only allows pair-wise comparison between two distributions, JSD enables a more comprehensive assessment of the diversity within a population by accommodating multiple distributions. Meanwhile, JSD is symmetrically defined, which is invariant under the interchange of the distributions being compared, helping simplify the implementation.

Despite the effectiveness of the population-based training with the $\mathcal{L}_{\text{div}}$ in Eqn.~\ref{loss_div}, the continual generation would still result in teammate groups with similar behaviors in different generations without other guarantees. Meanwhile, the size of the population $n_p$ might also have a significant impact. Inspired by the relationship between similarity and compatibility proved in ~\cite{lipo}, we extend the theorem to our CT-Dec-POMDP:

\begin{definition}[$\epsilon$-compatible teammates]\label{compatibility}
For the controllable agents $\boldsymbol{\pi}_{\text{ego}}$, let $\mathcal{J}_{\text{sp}}(\boldsymbol{\pi}_{\text{ego}})=\alpha$. We refer $\boldsymbol{\pi}_{\text{tm}}$ as an $\epsilon$-compatible teammate  $\boldsymbol{\pi}_{\text{ego}}$ if and only if $\mathcal{J}(\langle\boldsymbol{\pi}_{\text{ego}}, \boldsymbol{\pi}_{\text{tm}}\rangle)\geq (1-\epsilon)\alpha$.
\end{definition}
\begin{theorem}
    Given the controllable agents $\boldsymbol{\pi}_{\text{ego}}$ and teammate policies $\boldsymbol{\pi}_{\text{tm}}$ and $\forall \boldsymbol{\pi}_{\text{tm}}'$, $\boldsymbol{\pi}_{\text{tm}}, \boldsymbol{\pi}_{\text{tm}}'$ are $\epsilon-$similar policies. Then we have
    $(1-\epsilon)\mathcal{J}(\langle\boldsymbol{\pi}_{\text{ego}},\boldsymbol{\pi}_{\text{tm}})\leq\mathcal{J}(\langle\boldsymbol{\pi}_{\text{ego}}, \boldsymbol{\pi}_{\text{tm}}'\rangle)\leq (1+\epsilon)\mathcal{J}(\langle\boldsymbol{\pi}_{\text{ego}},\boldsymbol{\pi}_{\text{tm}}\rangle)$.
    \label{thm2}
\end{theorem}

The underlying idea behind Thm.~\ref{thm2} is that controllable agents, when effectively collaborating with a specific teammate group, will also be compatible with the teammate group's $\epsilon$-similar policies. Proofs are given in App.~\ref{ProofsforTheorems}. We thus have the following corollary:

\begin{corollary}
    Given the controllable agents $\boldsymbol{\pi}_{\text{ego}}$ and teammates $\boldsymbol{\pi}_{\text{tm}}$. If $\mathcal{J}(\langle\boldsymbol{\pi}_{\text{ego}}, \boldsymbol{\pi}_{\text{tm}}'\rangle)<(1-\epsilon)\mathcal{J}(\langle\boldsymbol{\pi}_{\text{ego}},\boldsymbol{\pi}_{\text{tm}})$, then $\boldsymbol{\pi}_{\text{tm}}$ and $\boldsymbol{\pi}_{\text{tm}}'$ are not $\epsilon$-similar policies, i.e., $d(\boldsymbol{\pi}_{\text{tm}}, \boldsymbol{\pi}_{\text{tm}}')>\epsilon$.
    \label{coro2}
\end{corollary}

The result from Cor.~\ref{coro2} shows that we can ensure that teammate groups generated in the current population are different from those before by decreasing its compatibility with the controllable agents $\boldsymbol{\pi}_{\text{ego}}$, which are trained to effectively collaborate with the teammates generated so far. Assuming that the controllable agents are fixed during the teammate population evolving stage, the optimization objective can be written as:
\begin{equation}
    \begin{aligned}
        \mathcal{L}_{\text{incom}} = -\frac{1}{n_p}\sum_{i=1}^{n_p} \mathcal{J}(\langle\boldsymbol{\pi}_{\text{ego}}, \boldsymbol{\pi}_{\text{tm}}^i\rangle).
    \end{aligned}
    \label{loss_incom}
\end{equation}

To ensure the meaningful learning of teammate groups' policies, it is crucial for each individual in the population to be capable of cooperating with complementary policies. Thus, the optimization of teammate focuses on maximizing the following objective:
\begin{equation}
    \begin{aligned}
        \mathcal{L}_{\text{sp}} = \frac{1}{n_p}\sum_{i=1}^{n_p} \mathcal{J}_{\text{sp}}( \boldsymbol{\pi}_{\text{tm}}^i),~
        where~i=1, 2,.., n_p.
    \end{aligned}
    \label{loss_sp}
\end{equation}

Considering the specified objectives, the complete objective function for the teammate population is as follows:
\begin{equation}
    \begin{aligned}
        \mathcal{L}_{\text{tm}} = \mathcal{L}_{\text{sp}} + \alpha_{\text{div}}  \mathcal{L}_{\text{div}} + \alpha_{\text{incom}}  \mathcal{L}_{\text{incom}},
    \end{aligned}
    \label{loss_tm}
\end{equation}
where, $\alpha_{\text{div}}$ and $\alpha_{\text{incom}}$ are adjustable hyper-parameters that control the balance between the three objectives.

\subsection{Compatible Coordination Policy Learning} \label{marlpolicylearning}

After generating a new teammate population that is diverse and incompatible with the controllable agents, we aim to train the controllable agents to effectively cooperate with newly generated teammate groups, as well as maintain the coordination ability with the trained ones. It requires the controllable agents to possess the continual learning ability, as introduced in Sec.~\ref{problem formulation}, where teammate policies appear sequentially in CT-Dec-POMDP.

In the context of evolutionary-generated teammate groups appearing sequentially, employing a single generalized policy network poses challenges due to the existence of multi-modality and varying behaviors among teammate groups. Consequently, conflicts and degeneration in the controllable agents' policies may arise. To address this issue, recent approaches like MACPro~\cite{yuan2023multi} have adopted a solution where customized heads are learned for each specific task.
Building upon this idea, our approach involves designing a policy network with a shared backbone denoted as $f_{\phi}$, complemented by multiple output heads represented as $\{h_{\psi_i}\}_{i=1}^m$. The shared backbone is responsible for extracting relevant features, while each output head handles making the final decisions.

With the structured policy network, when paired with the new teammate group's policy $\boldsymbol{\pi}_{\text{tm}}^{k+1}$, we first instantiate a new output head $h_{\psi_{m+1}}$. Subsequently, our focus shifts to training the controllable agents to effectively cooperate with the new teammate group.
\begin{equation}
    \begin{aligned}
        \mathcal{L}_{\text{com}} = \mathcal{J}(\langle \boldsymbol{\pi}_{\text{ego}}, \boldsymbol{\pi}_{\text{tm}}^{k+1} \rangle).
    \end{aligned}
    \label{loss_com}
\end{equation}
It is worth noting that once trained, the output heads $\{h_{\psi_i}\}_{i=1}^m$ remain fixed, and during the training process, the gradient $\mathcal{L}_{\text{com}}$ only propagates through the parameters $\phi$ and $\psi_{m+1}$.

Training the best response via $\mathcal{L}_{\text{com}}$ enables us to derive a policy that is capable of cooperating with the new teammate group $\boldsymbol{\pi}_{\text{tm}}^{k+1}$. However, the use of one shared backbone poses a challenge as it inevitably leads to forgetting previously learned cooperation, especially when encountering teammates with different behaviors, resulting in failure to cooperate with teammates seen before.
One straightforward approach to address this issue is to fix the parameters of the backbone upon completing the training of the first policy head. However, this approach has obvious drawbacks. On one hand, the fixed backbone might fail to extract common features adequately due to the limited coverage of training data. On the other hand, the output head's capacity might be insufficient, leading to suboptimal performance when training to cooperate with new teammates.
To mitigate the problem of catastrophic forgetting and enhance the policy's expressiveness, we apply a regularization objective by constraining the parameters from changing abruptly while learning the new output head $h_{\psi_{m+1}}$:

\begin{equation}
    \begin{aligned}
        \mathcal{L}_{\text{reg}} = \frac{1}{m}\sum_{i=1}^m ||\phi-\phi_i||_p,
    \end{aligned}
    \label{loss_reg}
\end{equation}
where $\phi_i$ is the saved snapshot of the backbone $\phi$ after obtaining the $i^{\text{th}}$ output head, and $||\cdot||_p$ is $l_{p}$ norm. This regularization mechanism helps to retain previously learned knowledge and ensures that the shared backbone adapts to the new teammate. Striking a balance between adaptability and retaining relevant knowledge, we can effectively enhance the cooperative performance of the policy with diverse teammates.
The overall objective of the controllable agents when encountering the $(k+1)^{\text{th}}$ teammate group is defined as:
\begin{equation}
    \begin{aligned}
        \mathcal{L}_{\text{ego}} = \mathcal{L}_{\text{com}} + \alpha_{\text{reg}}\mathcal{L}_{\text{reg}}
    \end{aligned},
    \label{loss_ego}
\end{equation}
where $\alpha_{\text{reg}}$ is a tunable weight.

Despite the effectiveness of combining the proposed $\mathcal{L}_{\text{ego}}$ and the carefully designed policy network architecture, a major limitation lies in its poor scalability as the number of output heads increases linearly with the dynamically generated teammate groups. To address this limitation and achieve better scalability, we propose a resilient head expansion strategy that effectively reduces the number of output heads while maintaining the policy's compatibility:
\begin{itemize}
    \item Upon completing the training of the output head $h_{\psi_{m+1}}$, we proceed to evaluate the coordination performance of this head and all the existing ones $\{h_{\psi_{i}}\}_{i=1}^{m+1}$ when paired with the new teammate group's policy $\boldsymbol{\pi}_{\text{tm}}^{k+1}$. The coordination performance is measured using the empirical average return $\{\hat{R}_i\}_{i=1}^{m+1}$, where $\hat{R}_i=\frac{1}{N} \sum_{j=1}^N G(\tau_j^i)$ represents the average return obtained by executing trajectories $\tau_j^i$ generated by applying the $i^{\text{th}}$ output head.
    \item To manage the number of output heads and prevent uncontrolled growth, we choose to retain the newly trained head if its performance surpasses a certain threshold compared to the best-performing existing head. Formally, we keep the newly trained head if $\frac{\hat R_{m+1}-\max_{i} \{\hat R_i\}_{i=1}^m}{\max_{i} \{\hat R_i\}_{i=1}^m}\geq \lambda$. This approach ensures that we only expand the number of output heads when there is a substantial improvement in performance, indicating that the new teammate group's behavior requires a distinct policy. Otherwise, if the existing output heads are sufficiently generalized to cooperate effectively with the new teammate, no new head will be expanded.
\end{itemize}

By adopting this resilient head expand strategy, we strike a balance between reducing the number of output heads and maintaining the policy's adaptability, resulting in a more scalable and efficient approach to handling dynamic teammate groups under the continual coordination setting.

\subsection{Overall Algorithm} \label{overall algorithm}
In this section, we present a comprehensive overview of the Macop (Multi-agent Compatible Policy Learning) procedure. Macop aims to train controllable agents to effectively cooperate with various teammate groups. 
During the training phase, Macop employs an evolutionary method to generate diverse and incompatible teammate groups and trains the controllable agents to be compatible with the teammates under the continual setting.
In each iteration (generation) $k (k > 1)$, we first select the $(k-1)^{\text{th}}$ teammate population $\mathcal{P}_{\text{tm}}^{k-1}$ as the parent population. Then, the offspring population is derived by training the parent population with $\mathcal{L}_{\text{tm}}$ in Eqn.~\ref{loss_tm}, i.e., mutation. The teammate groups are constructed based on value-based methods~\cite{vdn, qmix}, and $\boldsymbol{\pi}_{\text{tm}}(\cdot|s)$ is replaced with $softmax(Q_{\text{tm}}^i(\cdot|s))$ in $\mathcal{L}_{\text{div}}$ for practical use. With $n_p$ teammate groups of the parent population and $n_p$ teammate groups of the offspring population, we apply a carefully designed selection scheme as follows. To expedite the training of meaningful teammate groups, we first eliminate $\lfloor\frac{n_p}{2} \rfloor$ teammate groups with the lowest self-play return, i.e., $\max_{\bar{\boldsymbol{\pi}}_{\text{ego}}^i\in\Pi_{\text{tm}}^c}\mathcal{J}(\langle\bar{\boldsymbol{\pi}}_{\text{ego}}^i, \boldsymbol{\pi}_{\text{tm}}^i\rangle)$. Next, we proceed to eliminate $\lceil\frac{n_p}{2} \rceil$ teammate groups with the highest cross-play return under the controllable agents, i.e., $\mathcal{J}({\langle\boldsymbol{\pi}_{\text{ego}}, \boldsymbol{\pi}_{\text{tm}}^i\rangle})$, so as to improve incompatibility. Finally, we utilize the remaining $n_p$ teammate groups as the new teammate population of iteration $k$, i.e., $\mathcal{P}_{\text{tm}}^{k}$. 

\begin{figure*}
\setlength{\abovecaptionskip}{0cm}
  \centering
    \subfigure[Level-based Foraging (LBF)]{
    \label{env lbf4}
    \includegraphics[width=0.23\textwidth]{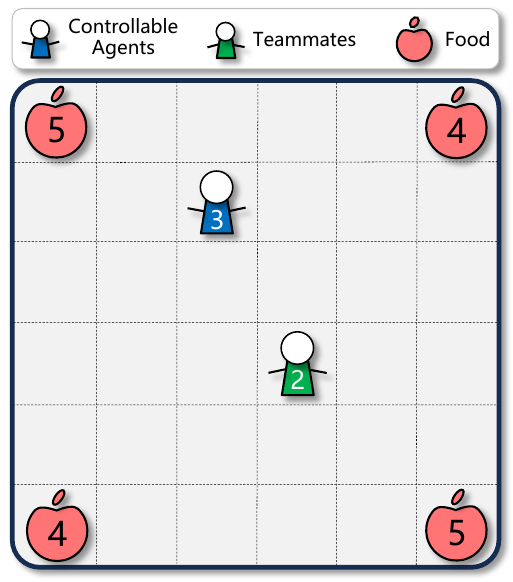}
    }
    \subfigure[Predator Prey (PP)]{
    \label{env pp1}
    \includegraphics[width=0.23\textwidth]{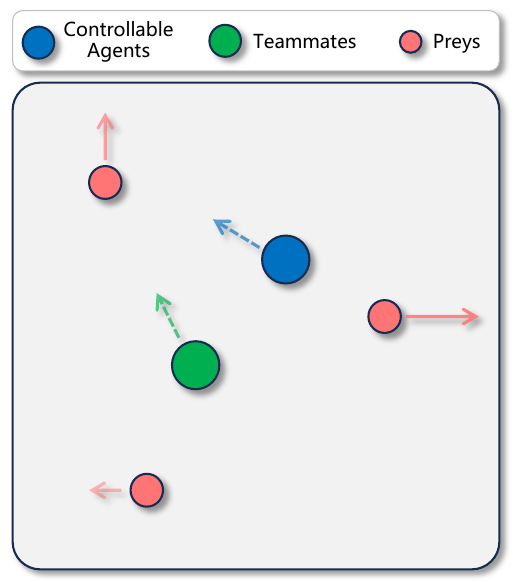}
    }
    \subfigure[Cooperative Navigation (CN)]{
    \label{env cn3}
    \includegraphics[width=0.23\textwidth]{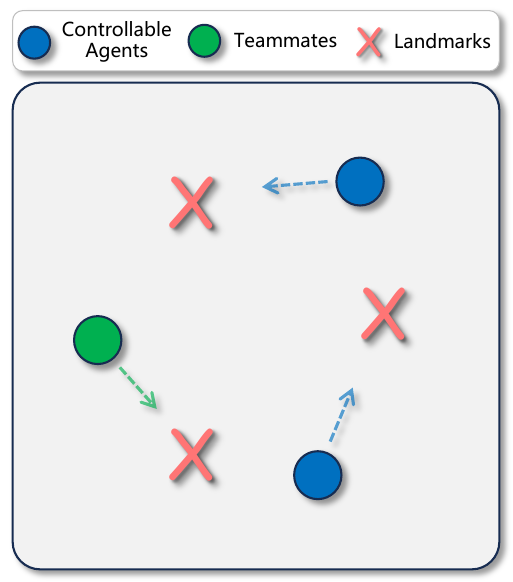}
    }
    \subfigure[Starcraft Multi-agent Challenge (SMAC)]{
    \label{env smac1}
    \includegraphics[width=0.23\textwidth]{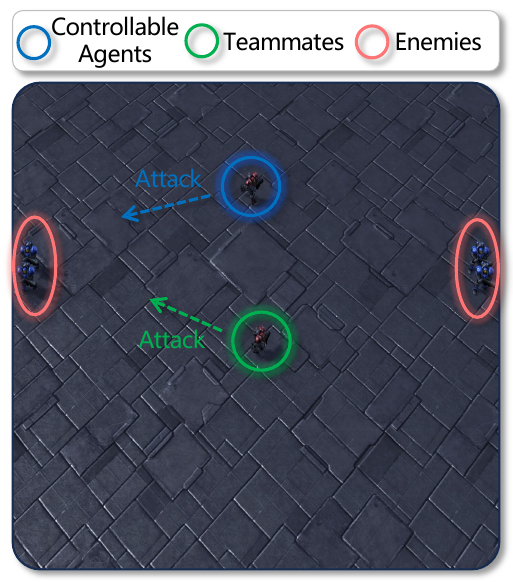}
    }
  \caption{Environments used in this paper, all details could be seen in App.~\ref{moreexperiments}. }
  \label{environments}
\end{figure*}

With the teammate population $\mathcal{P}_{\text{tm}}^{k}$ in place, we construct $n_p$ continual coordination processes in a sequential order and train controllable agents to learn compatible policies. The controllable agents are optimized using $\mathcal{L}_{\text{ego}}$ (defined in Eqn.~\ref{loss_ego}), and the output head is expanded as introduced in Sec.~\ref{marlpolicylearning}.

To determine when the continual process should be terminated, a carefully designed stopping criterion is employed. 
The training phase terminates at the $k^{\text{th}}$ iteration if the minimum cross-play return of $\mathcal{P}_{\text{tm}}^{k+1}$ and the controllable agents in iteration $k$ exceeds a certain value, i.e., $C = \frac{\min_{i}\mathcal{J}(\langle\boldsymbol{\pi}_{\text{ego}},\boldsymbol{\pi}_{\text{tm}}^i \rangle)}{\frac{1}{n_p}\sum_{i=1}^{n_p} \mathcal{J}_{\text{sp}}( \boldsymbol{\pi}_{\text{tm}}^i)}\geq \xi , ~ \boldsymbol{\pi}_{\text{tm}}^i \in \mathcal{P}_{\text{tm}}^{k+1}$. It indicates that the controllable agents at the $k^{\text{th}}$ iteration can effectively cooperate with the $(k+1)^{\text{th}}$ teammate population even they have been trained to decrease the compatibility, and the teammate policy space is covered for a given environment.

During the testing phase, a meta-testing paradigm is employed to determine which output head is selected to pair with an unknown teammate group. Initially, all output heads are allowed to interact with the teammate group to collect a few trajectories, and their cooperation abilities are evaluated based on empirical returns. The output head with the highest performance is then chosen for testing.
The pseudo-codes for both the training and testing phases of our Macop procedure are provided in App.~\ref{macopdetails}.

\section{Experiments}
In this section, we conduct a series of experiments to answer the following questions: 1) Can Macop generate controllable agents capable of effectively collaborating with diverse teammates in different scenarios, surpassing the performance of other methods?
2) Does the evolutionary generation of teammates bring about a noticeable increase in diversity, and how do our controllable agents compare to other baseline models in terms of compatibility?
3) What is the detailed training process of Macop?
4) How does each component and hyperparameter influence Macop? 

\subsection{Environments and Baselines}

We here select four multi-agent coordination environments and design eight scenarios as evaluation benchmarks (Fig.~\ref{environments}). Level-based Foraging (LBF)~\cite{lbf} presents a challenging multi-agent cooperative game, where agents with varying levels navigate through a grid world, collaboratively striving to collect food with different levels. The successful collection occurs when the sum of levels of participating agents matches or exceeds the level of the food item. Predator Prey (PP) and Cooperative Navigation (CN) are two benchmarks coming from the popular MPE environment~\cite{maddpg}.
In the PP scenario, agents (predators) must together pursue the moving adversaries (prey). On the other hand, in CN, multiple agents receive rewards when they navigate toward landmarks while ensuring they avoid collisions with one another. 
We also conduct experiments in the widely used StarCraft II combat scenario, SMAC~\cite{pymarl}, which involves unit micromanagement tasks. In this setting, ally units are trained to beat enemy units controlled by the built-in AI. We specifically design two scenarios for each mentioned benchmark (e.g., PP1 and PP2), and details could be found in App.~\ref{moreexperiments}.

\begin{table*}
    \centering
    \caption{Average test return $\pm$ std when paired with teammate groups from \textit{evaluation set} in different scenarios. We re-scale the value by taking the result of Finetune as an anchor and present average performance improvement w.r.t Finetune. The best result of each column is highlighted in \textbf{bold}. The symbols `$+$', `$\approx$', and `-' indicate that the result is significantly inferior to, almost equivalent to, and superior to Macop, respectively, based on the Wilcoxon rank-sum test~\cite{mann1947test} with confidence level $0.05$.}
    \resizebox{0.99\textwidth}{!}{\begin{tabular}{c|cccccccc|c}
        \toprule
        \multirow{2}{*}{\diagbox{Method}{Env}} & \multicolumn{2}{c}{LBF} & \multicolumn{2}{c}{PP} & \multicolumn{2}{c}{CN} & \multicolumn{2}{c|}{SMAC} & Avg. Performance \\
        \cmidrule(lr){2-3} \cmidrule(lr){4-5} \cmidrule(lr){6-7} \cmidrule(lr){8-9}
        & LBF1 & LBF4 & PP1 & PP2 & CN2 & CN3 & SMAC1 & SMAC2 & Improvement (\%)\\
        \midrule
        Macop (ours) & $1.14 \pm 0.02$ & $\mathbf{1.64 \pm 0.03}$ & $\mathbf{1.73 \pm 0.11}$ & $\mathbf{2.14 \pm 0.53}$ & $\mathbf{1.66 \pm 0.03}$ & $\mathbf{1.70 \pm 0.06}$ & $\mathbf{1.26 \pm 0.42}$ & $1.56 \pm 0.17$ & $\mathbf{60.44}$\\
        Single Head & $0.98 \pm 0.07$ & $1.10 \pm 0.32$ & $0.87 \pm 0.58$ & $1.44 \pm 0.52$ & $1.01 \pm 0.49$ & $0.99 \pm 0.24$ & $1.06 \pm 0.14$ & $1.25 \pm 0.40$  & $8.92$\\
        Random Head & $0.92 \pm 0.05$ & $0.85 \pm 0.10$ & $0.88 \pm 0.17$ & $1.18 \pm 0.39$ & $0.98 \pm 0.23$ & $0.92 \pm 0.11$ & $0.97 \pm 0.14$ & $1.28 \pm 0.21$ & $-0.25$\\
        LIPO~\cite{lipo} & $1.07 \pm 0.09$ & $1.53 \pm 0.14$ & $1.64 \pm 0.21$ & $1.93 \pm 0.52$ & $1.13 \pm 0.41$ & $1.33 \pm 0.25$ & $1.19 \pm 0.18$ & $1.08 \pm 0.21$ & $36.27$\\
        FCP~\cite{strouse2021collaborating} & $\mathbf{1.16 \pm 0.02}$ & $1.33 \pm 0.06$ & $1.17 \pm 0.85$ & $1.34 \pm 0.12$ & $0.90 \pm 0.48$ & $1.41 \pm 0.23$ & $0.97 \pm 0.19$ & $1.54 \pm 0.10$ & $25.82$\\
        TrajeDi~\cite{lupu2021trajectory} & $1.16 \pm 0.06$ & $1.34 \pm 0.11$ & $1.68 \pm 0.33$ & $1.56 \pm 0.52$ & $1.29 \pm 0.23$ & $1.53 \pm 0.11$ & $1.25 \pm 0.12$ & $\mathbf{1.57 \pm 0.16}$ & $42.26$\\    
        EWC~\cite{kirkpatrick2017overcoming} & $0.97 \pm 0.08$ & $0.99 \pm 0.16$ & $0.83 \pm 0.48$ & $0.77 \pm 0.43$ & $0.57 \pm 0.37$ & $0.71 \pm 0.27$ & $1.03 \pm 0.13$ & $0.61 \pm 0.09$ & $-18.82$\\
        Finetune & $1.00 \pm 0.16$ & $1.00 \pm 0.27$ & $1.00 \pm 0.58$ & $1.00 \pm 0.68$ & $1.00 \pm 0.31$ & $1.00 \pm 0.24$ & $1.00 \pm 0.17$ & $1.00 \pm 0.23$ & $/$\\
        \midrule
        $+/\approx/-$ & $3/4/0$& $6/1/0$& $2/5/0$ &$6/1/0$ &$7/0/0$ &$7/0/0$ &$5/2/0$ &$4/3/0$ & $7/0/0$ \\ 
        \bottomrule
    \end{tabular}}
    \label{table1}
\end{table*}

To investigate whether Macop is capable of coordinating with diverse seen/unseen teammates, we implement Macop on the popular value-based methods VDN~\cite{vdn} and QMIX~\cite{qmix}, and compare it with multiple baselines. 
First, to assess the impact of the teammate generation process on the coordination ability of the controllable agents, we compare Macop with FCP~\cite{strouse2021collaborating}, which initially generates a set of teammate policies independently and then trains the controllable agents to be the best response to the set of teammates. The diversity among teammate polices is achieved solely through network random initialization.
Additionally, we examine another population-based training mechanism that trains the teammate population using both $\mathcal{L}_{\text{sp}}$ and $\mathcal{L}_{\text{div}}$, aiming to generate teammates with enhanced diversity. This approach, which aligns with existing literature~\cite{lupu2021trajectory, ding2023coordination}, is referred to as TrajeDi for convenience.
On the other hand, LIPO~\cite{lipo} induces teammate diversity by reducing the compatibility between the teammate policies in the population. Concretely, it trains the teammate population with an auxiliary objective $\mathcal{J}_{\text{LIPO}} = -\sum_{i\neq j}\mathcal{J}(\langle  \boldsymbol{\pi}_{\text{tm}}^i, \boldsymbol{\pi}_{\text{tm}}^j \rangle)$, where the indices $i, j$ refer to two randomly sampled teammates in the population.
Furthermore,  with the teammate generation module held constant, we proceed to compare Macop with Finetune. Finetune directly tunes all the parameters of the controllable agents to coordinate with the currently paired teammate group. We also investigate two other approaches: Single Head, which applies regularization $\mathcal{L}_{\text{reg}}$ to the backbone but does not utilize the multi-head architecture, and Random Head, which randomly selects an existing head during evaluation, thus verifying the necessity of Macop's testing paradigm.
Finally, we employ the popular continual learning method EWC~\cite{kirkpatrick2017overcoming} to learn to coordinate with the teammates generated by TrajeDi, thereby providing an overall validation of the effectiveness of Macop.
More details  are illustrated in App.~\ref{moreexperiments}.


\subsection{Competitive Results}

In this section, we analyze the effectiveness of the controllable agents learned from different methods from two aspects: coordination performance with diverse seen/unseen teammates, and continual learning ability on a sequence of incoming teammates.

\textbf{Overall Coordination Performance}~To ensure a fair comparison of coordination performance, we aggregate all the teammate groups generated by Macop and baselines into an \textit{evaluation set}. For each method, we pair the learned controllable agents with teammate groups in this \textit{evaluation set} to run 32 episodes for each pairing. The average episodic return over all episodes when pairing with different teammate groups is calculated as the evaluation metric. This metric serves as a comprehensive measure of the overall coordination performance and generalization ability of the controllable agents. We run each method for five distinct random seeds.

As depicted in Tab.~\ref{table1}, we observe that approaches such as FCP, TrajeDi, and LIPO exhibit limited coordination generalization ability in different scenarios, especially when the population size is restricted. This highlights the need for ample coverage in teammate policy space to establish a robust coordination policy. Intriguingly, among the three methods mentioned, we found no significant differences, indicating that certain design elements, such as instance diversity among teammates, fail to fundamentally address this challenge.
In contrast, when using generated teammates, simply finetuning the multi-agent policy or employing widely-used continual approaches like EWC exhibits inferior coordination performance, as confirmed by our experiments and in line with the findings in MACPro~\cite{yuan2023multi}. This suggests that specialized designs tailored for multi-agent continual settings play a crucial role.
On the other hand, Macop exhibits a remarkable performance advantage over nearly all baselines across various scenarios, demonstrating that controllable agents trained by Macop possess robust coordination abilities. Furthermore, we discovered that the Single Head architecture struggles due to the presence of multi-modality in teammate behavior, underscoring the necessity of a multi-head architecture. An effectively designed testing paradigm, utilizing multiple available learned heads, proves indispensable. It is worth noting that Random Head fails to select the optimal head for evaluation, resulting in a degradation in performance.
Our pipeline relies on efficient design for continual learning, and more comprehensive results on the necessity of each component can be found in Sec.~\ref{ablation and sensitivity}.

\begin{table*}[htbp]
\centering
\caption{Continual Learning Ability. Average BWT/FWT $\pm$ std of four different methods in different evaluated environments. }
\resizebox{0.99\textwidth}{!}{
\begin{tabular}{c|cccccccccccc}
\toprule
\multirow{2}{*}{Method} & \multicolumn{2}{c}{LBF4} & \multicolumn{2}{c}{PP1} & \multicolumn{2}{c}{CN3} & \multicolumn{2}{c}{SMAC1} \\
\cmidrule(lr){2-3} \cmidrule(lr){4-5} \cmidrule(lr){6-7} \cmidrule(lr){8-9}
& BWT & FWT & BWT & FWT & BWT & FWT & BWT & FWT \\
\midrule
Macop & $\mathbf{-0.01 \pm 0.02}$ & $\mathbf{0.07 \pm 0.07}$ & $\mathbf{0.03 \pm 0.04}$ & $-0.16 \pm 0.18$ & $\mathbf{0.04 \pm 0.06}$ & $\mathbf{0.10 \pm 0.09}$ & $\mathbf{-0.02\pm0.11} $ & $0.07 \pm 0.19$ \\
CLEAR & $-0.05 \pm 0.07$ & $\mathbf{0.07 \pm 0.06}$ & $0.01 \pm 0.08$ & $\mathbf{-0.05 \pm 0.11}$ & $-0.16 \pm 0.15$ & $0.00 \pm 0.20$ & $-0.50\pm 0.32$ & $0.04 \pm 0.35$ \\
EWC & $-0.30 \pm 0.08$ & $0.05 \pm 0.07$ & $-0.34 \pm 0.08$ & $-0.05 \pm 0.13$ & $-0.20 \pm 0.11$ & $0.03 \pm 0.11$ & $-1.02 \pm 0.47$ & $0.05 \pm 0.31$ \\
Finetune & $-0.34 \pm 0.07$ & $0.04 \pm 0.06$ & $-0.37 \pm 0.07$ & $-0.05 \pm 0.22$ & $-0.31 \pm 0.11$ & $0.05 \pm 0.10$ & $-1.24 \pm 0.51$ & $\mathbf{0.33 \pm 0.35}$ \\
\bottomrule
\end{tabular}}
\label{table2}
\end{table*}

\begin{figure*}
\setlength{\abovecaptionskip}{0cm}
  \centering
    \subfigure[]{
    \label{cn2_tm_scatter}
    \includegraphics[width=0.20\textwidth]{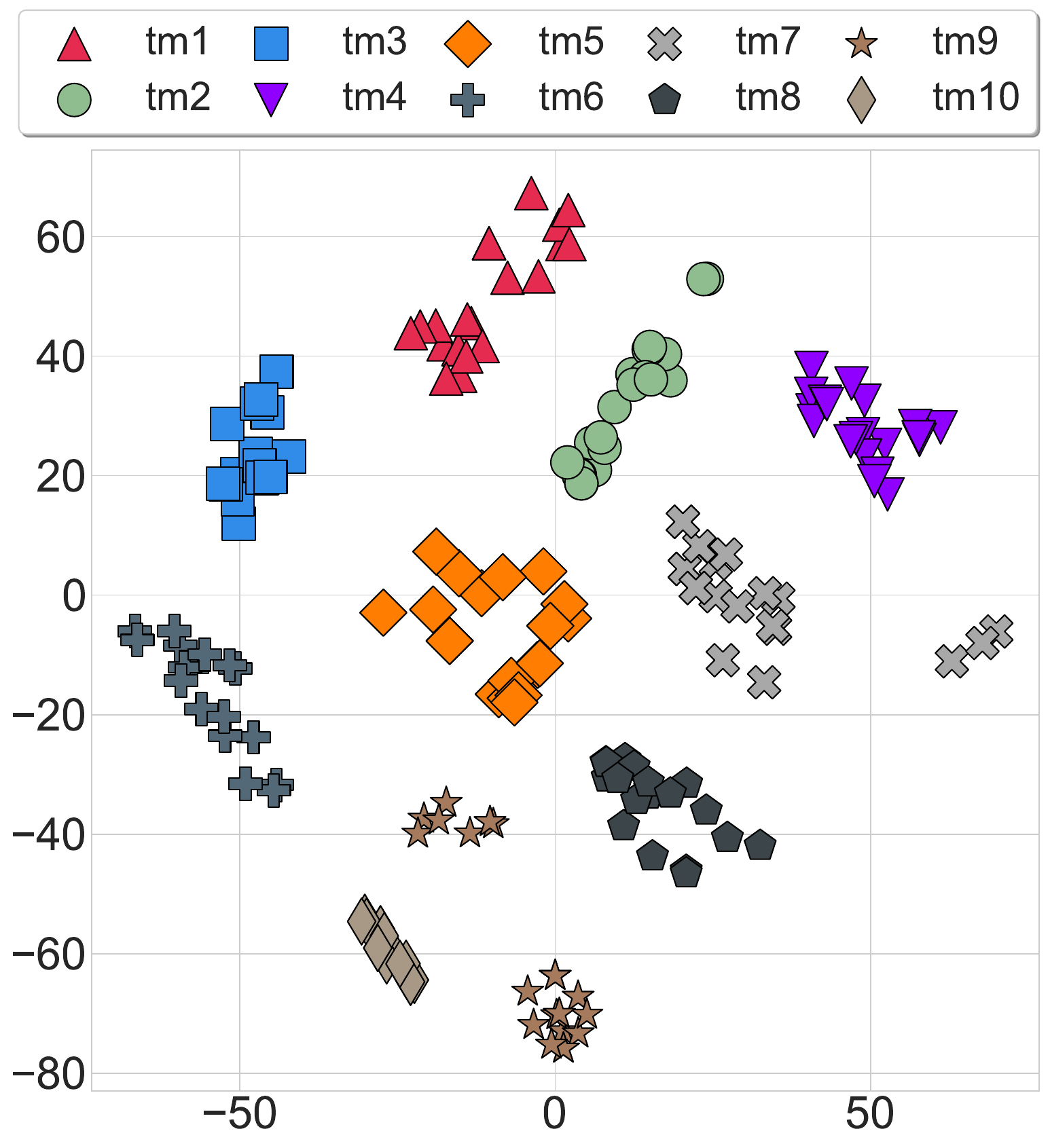}
    }
    \subfigure[]{
    \label{heatmap_ours}
    \includegraphics[width=0.23\textwidth]{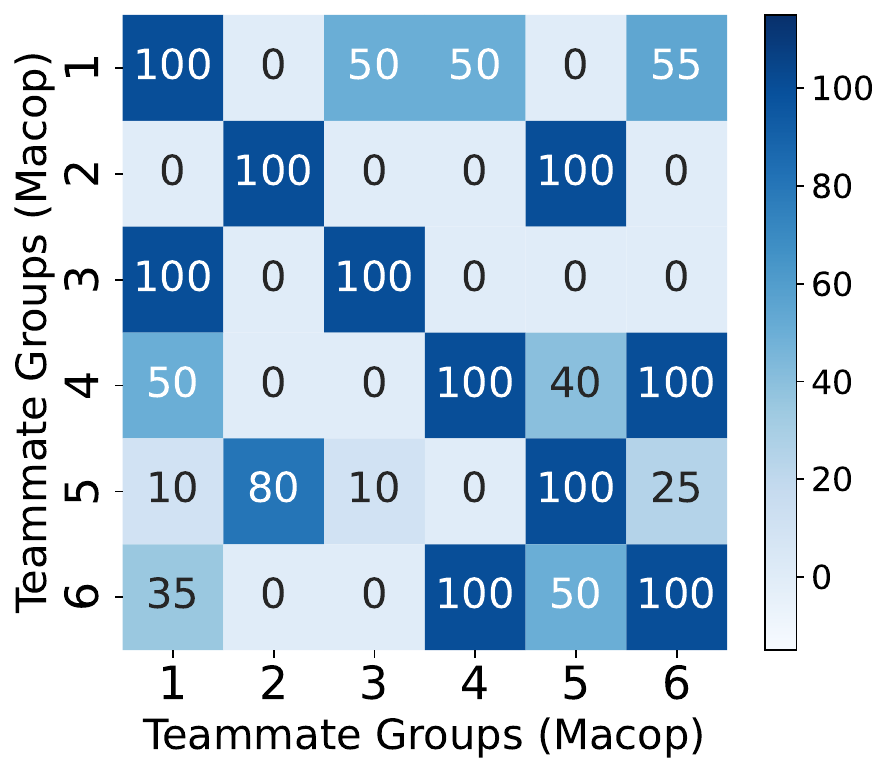}
    }
    \subfigure[]{
    \label{heatmap_trajedi}
    \includegraphics[width=0.23\textwidth]{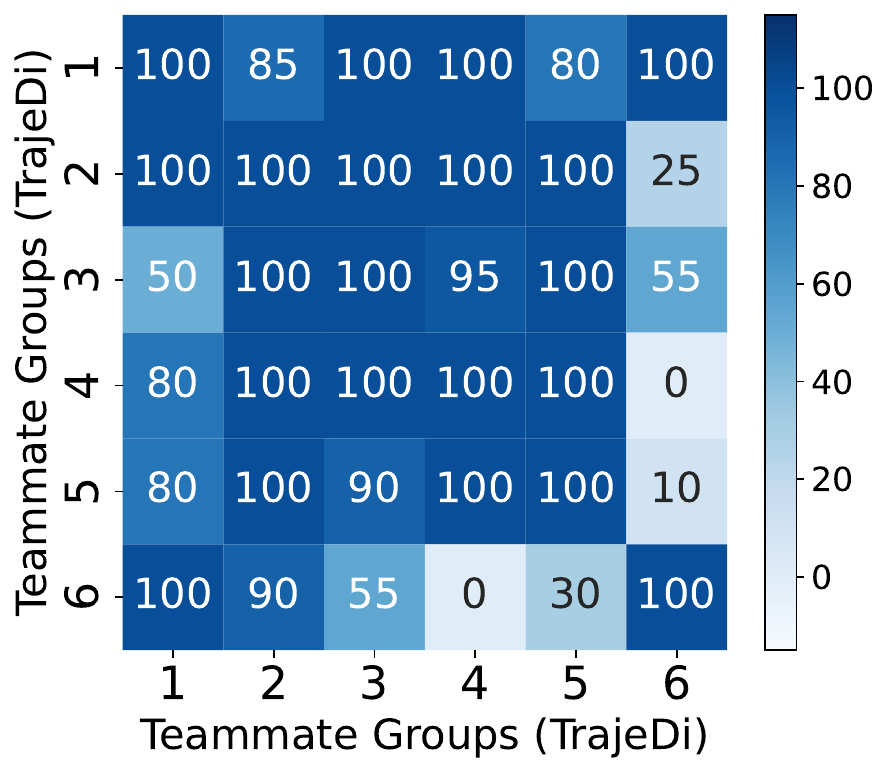}
    }
    \subfigure[]{
    \label{increase_tm}
    \includegraphics[width=0.27\textwidth]{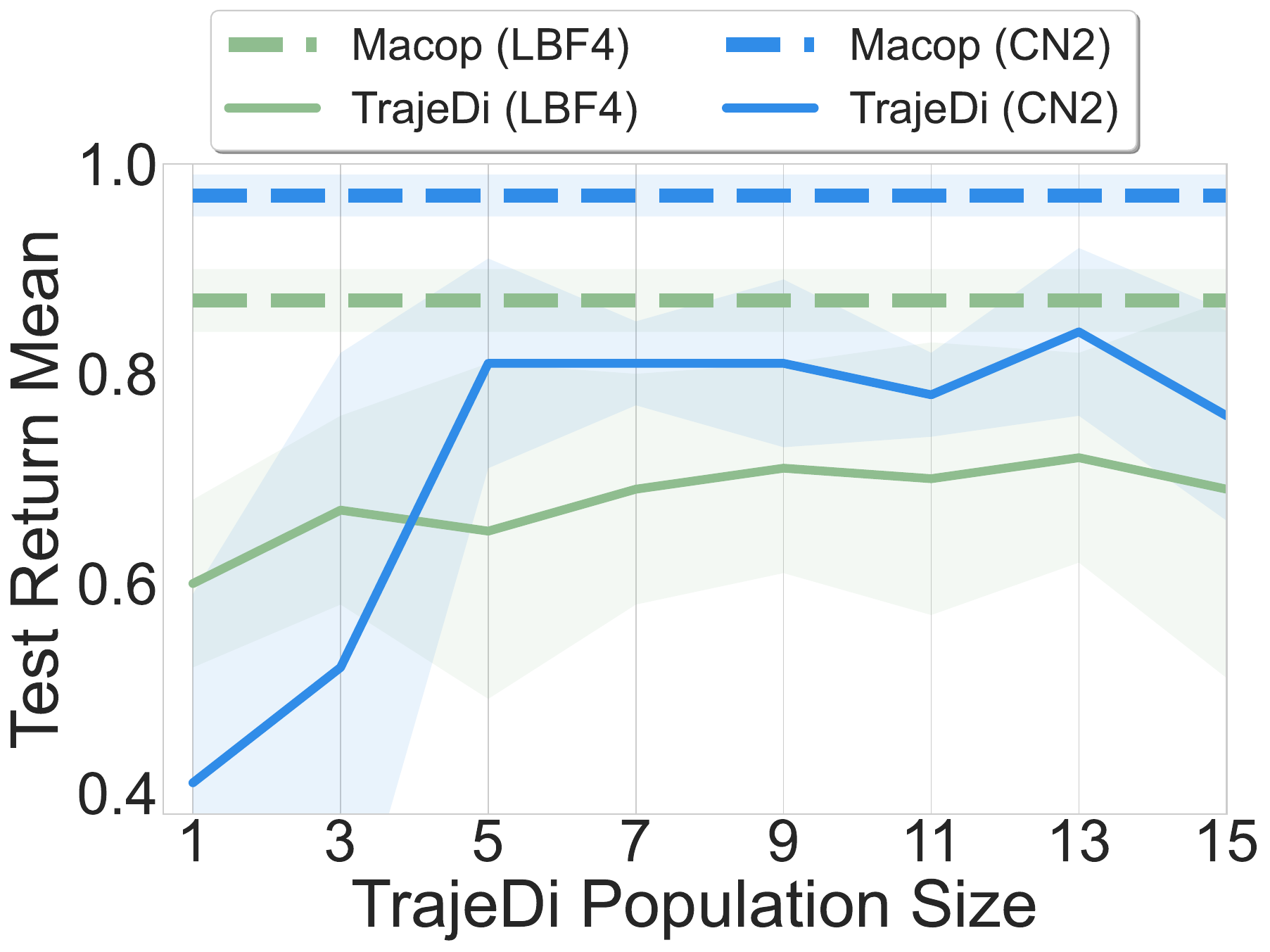}
    }
  \caption{Teammate policy space analysis. (a) The t-SNE projections of the self-play trajectory features of Macop's generated teammate groups in CN2. (b)(c) The cross-play returns of Macop's and TrajeDi's generated teammate groups in LBF4. (d) The change in TrajeDi's coordination ability with varying population sizes in LBF4 and CN2, compared with Macop.}
  \label{teammate space}
\end{figure*}

\textbf{Continual Learning Ability}~To further investigate the  continual learning ability of different methods, we utilize all teammate groups generated by Macop and all baselines to construct a fixed teammate sequence. Four continual learning methods are applied to train the controllable agents to coordinate with this teammate sequence in a continual manner, including Macop, CLEAR~\cite{Rolnick2018ExperienceRF}, EWC~\cite{kirkpatrick2017overcoming} and Finetune. CLEAR is a replay-based method which stores some data of previously trained teammates to rehearse the controllable agents when training with the current teammate group.
For a principled assessment of the continual learning ability, we introduce two metrics inspired by the concepts used in continual learning~\cite{wolczyk2021continual, wang2023comprehensive} within our CT-Dec-POMDP framework:

1) BWT$=\frac{1}{K-1}(\sum_{k=2}^{K}\frac{1}{k-1}\sum_{j=1}^{k-1}(\alpha_k^j -\alpha_j^j))$. 
BWT (Backward Transfer) evaluates the average influence of learning to cooperate with the newest teammate group on previously encountered teammates.
2) FWT$=\frac{1}{K-1}(\sum_{k=2}^{K}\frac{1}{k-1}\sum_{j=2}^k(\alpha_j^j-\tilde{\alpha}_j))$. FWT (Forward Transfer) assesses the average influence of all previously encountered teammate groups on the coordination performance of the new teammate.

Here, $\alpha_k^j$ represents the coordination performance of the controllable agents paired with the $j^{\text{th}}$ teammate group after training to cooperate with the $k^{\text{th}}$ teammate group, measured by the empirical episodic return. Additionally, $\tilde{\alpha}_j$ denotes the coordination performance of a randomly initialized complementary policy trained with the $j^{\text{{th}}}$ teammate group.

We record experimental results in Tab.~\ref{table2}. At first glance, Finetune demonstrates the worst BWT among all methods, validating the necessity of algorithm design to prevent catastrophic forgetting. However, even popular continual learning methods, CLEAR and EWC, grapple with forgetting to some degree. 
In contrast, Macop achieves the best BWT in all evaluated environments.
As for FWT, Macop obtains a competitive result compared with other methods.
Taking both BWT and FWT into consideration, Macop demonstrates a robust and adept continual learning ability. This aptitude empowers controllable agents to progressively acquire coordination proficiency with diverse teammates, and aligns seamlessly with the expanding coverage of the teammate policy space.


\subsection{Teammate Policy Space Analysis}


To investigate whether Macop is capable of generating teammate groups with diverse behaviors, a straightforward method involves comparing the self-play trajectories of different teammate groups.  Concretely, we first learn a transformer-based encoder to map trajectories into a low-dimensional feature space (details will be provided in App.~\ref{moreencoder}). We subsequently encode the teammates' self-play trajectories generated by Macop into the feature space. For visualization, we select 10 teammate groups from the CN2 scenario and extract their trajectory features, as shown in Fig.~\ref{cn2_tm_scatter}. The projection displays a notable dispersion, validating that teammate groups generated by Macop exhibit diverse behaviors as expected.

Furthermore, we conducted experiments to assess the compatibility among the generated teammate groups. In accordance with Def.~\ref{compatibility}, we paired different teammate groups in LBF4. The cross-play returns are presented in Fig.~\ref{teammate space}(b)(c), generated by  Macop and TrajeDi, respectively.
It is evident that when pairing two distinct groups from Macop, there is a noticeable drop in returns outside the main diagonal, indicating a lack of compatibility among the teammate groups generated by Macop. Conversely, the cross-play returns of TrajeDi's teammate groups are nearly identical to their self-play returns, suggesting a significantly lower level of incompatibility among teammate groups generated by TrajeDi because of poorer coverage of the teammate policy space.

To further explore whether methods without dynamically generating teammates can address policy space coverage by increasing the population size, we trained controllable agents using TrajeDi, varying in population size from 1 to 15. Subsequently, we evaluated the coordination ability using the \textit{evaluation set}, as depicted in Fig.~\ref{increase_tm}.
The results clearly illustrate that coordination ability improves as the population size increases until convergence is reached. However, a considerable performance gap between TrajeDi and Macop persists. Our analysis leads us to the conclusion that in intricate scenarios with multi-modality, vanilla methods that lack dynamic teammate generation struggle with new and unfamiliar teammates due to inadequate coverage of the teammates' policy space.  On the contrary, Macop's deliberate generation of incompatible teammates contributes to a more comprehensive coverage of the teammate policy space, ultimately enhancing its coordination ability.

\begin{figure*}
\setlength{\abovecaptionskip}{0cm}
  \centering
    \subfigure[]{
    \label{traj_4tm}
    \includegraphics[width=0.18\textwidth]{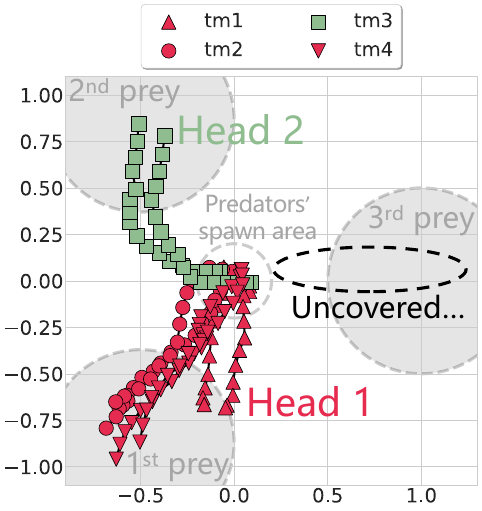}
    }
    \subfigure[]{
    \label{traj_5tm}
    \includegraphics[width=0.18\textwidth]{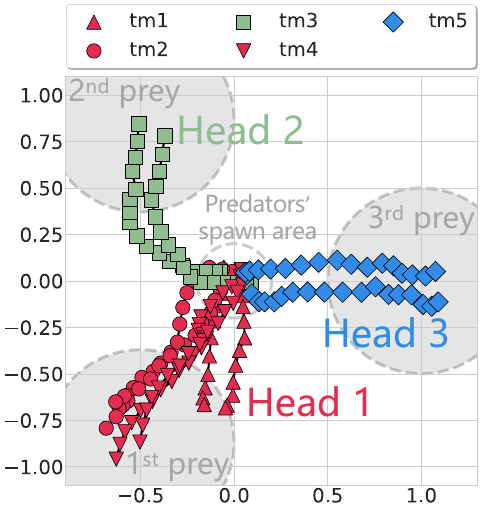}
    }
    \subfigure[]{
    \label{process_plot}
    \includegraphics[width=0.32\textwidth]{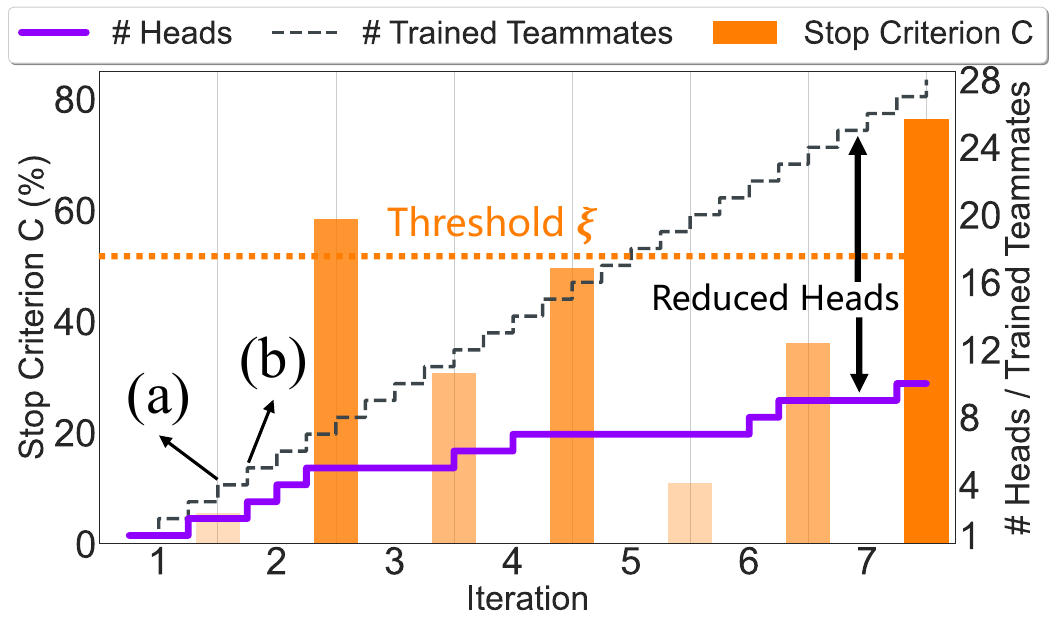}
    }
    \subfigure[]{
    \label{radar}
    \includegraphics[width=0.26\textwidth]{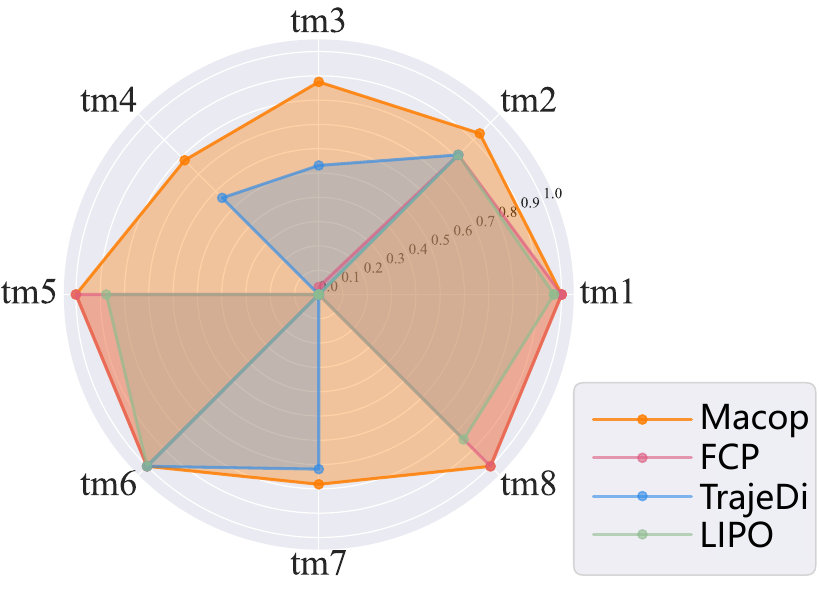}
    }
  \caption{Macop's learning process analysis. (a)(b) The self-play trajectories of the first four/five teammate groups. (c) The change of the number of trained teammate groups, the number of existing heads, and the stop criterion $C$ on each iteration. (d) Coordination performance comparison with different teammate groups in the \textit{evaluation set}.}
  \label{learning process analysis}
\end{figure*}

\subsection{Learning Process Analysis}

To gain a comprehensive understanding of Macop's functioning, it's essential to delve into its learning process, which involves generating incompatible teammates and refining controllable agents until convergence is achieved.
Fig.~\ref{learning process analysis} illustrates the process in PP1, showcasing key aspects, including the number of teammate groups generated, the number of existing heads, and the stop criterion $C$, all presented for each iteration (Fig.~\ref{process_plot}). 
In the first iteration, the teammate generation module produces a population of four distinct teammate groups, with three specializing in capturing the first prey and one focused on the second prey (Fig.~\ref{traj_4tm}). However, the population lacks desired diversity, as none of the groups learn to catch the remaining third prey. As for the controllable agents, they acquire the ability to collaborate with their teammates: Head 1 coordinates with those capturing the first prey, while Head 2 interacts with the group targeting the second prey.


During the second iteration, the teammate generation module generates new teammates incompatible with the controllable agents, expanding the coverage of the teammate policy space. As shown in Fig.~\ref{traj_5tm}, a new teammate group (identified as ``tm5'' in blue) successfully acquires the skill to capture the last prey, showcasing a completely novel behavior. Consequently, when the controllable agents complete their training with this new group, they establish a new head for better coordination.

The dynamic interplay between the adversarial teammate generation module and the training of controllable agents persists until the seventh iteration, resulting in an increased number of teammate groups and output heads. In this final iteration, the teammate generation module endeavors to generate seemingly ``incompatible'' teammates as it has throughout the training process, but it encounters failure. The generated teammate groups up to this point have already effectively covered a wide range of the teammate policy space. The controllable agents have successfully acquired the ability to coordinate with a sufficiently diverse array of teammates.
The newly generated teammate groups do not exhibit enough incompatibility, as indicated by the stop criterion surpassing the specified threshold $\xi$. This signifies that the cross-play performance between the controllable agents and these new ``incompatible'' teammates is comparable to the self-play performance of the teammate groups. It's worth noting that the $C$ value from the second iteration also exceeds the threshold, yet a minimum iteration count of 4 is enforced to ensure thorough exploration of the teammate policy space. This automated and self-regulating learning process within Macop concludes after the seventh iteration.
As a result of this process, Macop produces a notable set of 28 teammate groups with remarkable diversity, along with controllable agents that possess 10 heads. This is evidenced by their robust coordination abilities, which are prominently illustrated in Fig.~\ref{radar}.


\begin{figure}
  \centering
    \subfigure[Ablation Study]{
    \label{ablation}
    \includegraphics[width=0.33\textwidth]{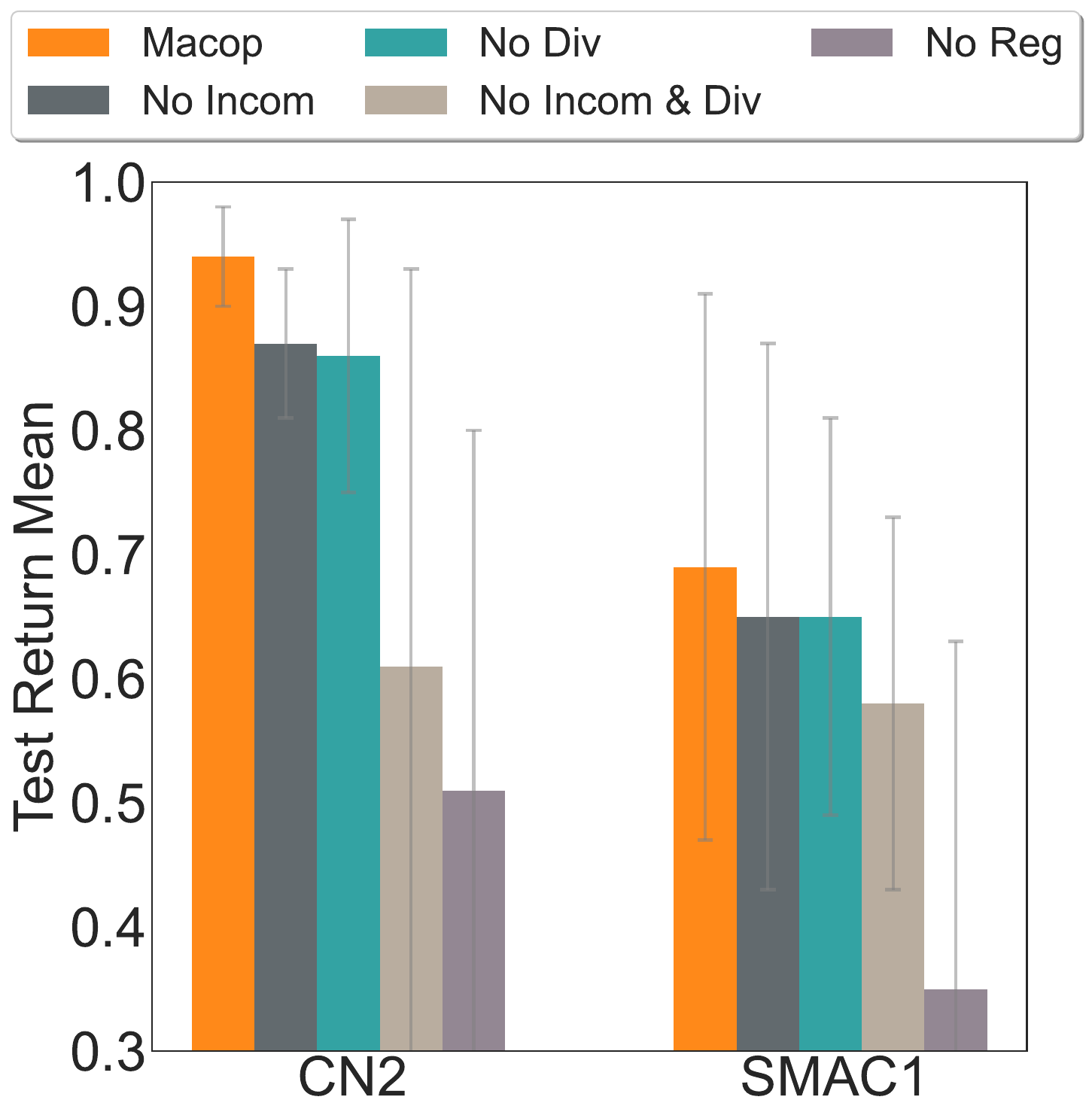}
    }
    \subfigure[Sensitivity of population size on PP2]{
    \label{sensitivity main}
    \includegraphics[width=0.31\textwidth]{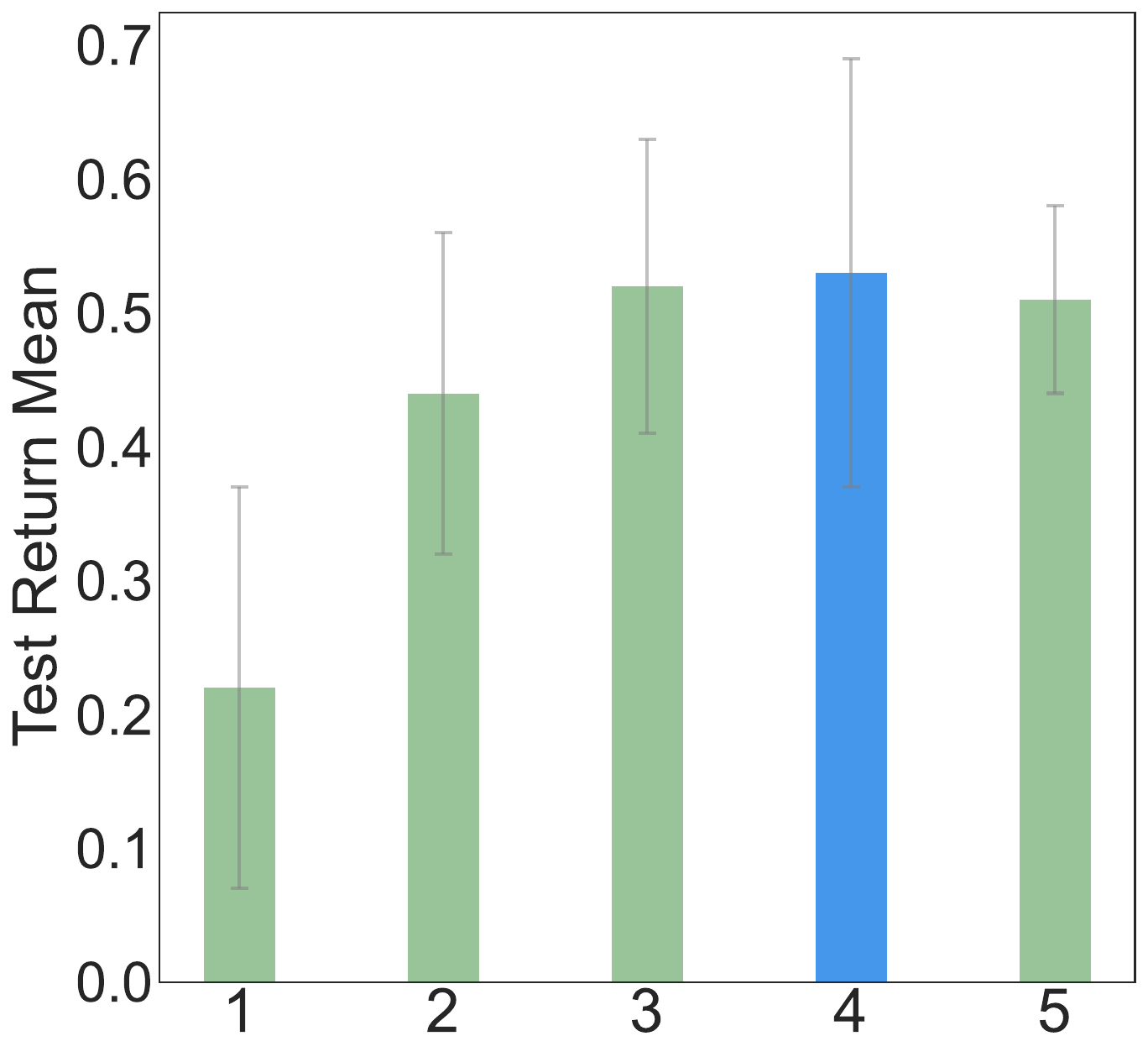}
    }
  \caption{Ablation and sensitivity studies.}
\end{figure}


\subsection{Ablation and Sensitivity Studies}\label{ablation and sensitivity}
We here conduct ablation studies on CN2 and SMAC1 to comprehensively assess the impacts of multiple modules. \textit{No Incom}, \textit{No Div}, and \textit{No Incom \& Div}, are derived by setting $\alpha_{\text{incom}} = 0$, $\alpha_{\text{div}} = 0$, and $\alpha_{\text{incom}} = \alpha_{\text{div}} = 0$, respectively.
Furthermore, we examine the impact of $\mathcal{L}_{\text{reg}}$, and designate this variant as \textit{No Reg} to explore the effects of regularization on the backbone network $\phi$.
To ensure a fair comparison, we incorporate the teammate groups generated by the four ablations into the \textit{evaluation set}. The results, as illustrated in Fig.~\ref{ablation}, reveal essential insights into the functioning of Macop. Removing $\mathcal{L}_{\text{incom}}$ or $\mathcal{L}_{\text{div}}$ leads to a performance degradation compared to the complete Macop, highlighting the significant contributions to the teammate diversity. Moreover, No Incom \& Div exhibits a substantial performance degradation, verifying the necessity of actively generating diverse teammates, instead of relying solely on random network initialization.
Furthermore, No Reg demonstrates the poorest performance among all the variants. The absence of regularization on the backbone network undermines the controllable agents' continual learning ability, weakening their coordination capability with diverse teammates.
These findings emphasize that each module plays an indispensable role in Macop.

As Macop includes multiple hyperparameters, we conduct experiments to investigate their sensitivity. The teammate groups generated by different hyperparameter settings are also incorporated into the \textit{evaluation set} for a fair comparison. One important hyperparameter is the population size $n_p$. On one hand, with a very small population, Macop cannot cover the teammate policy space in an efficient manner. On the other hand, setting the population size to an excessively large number will unnecessarily increase the running time of Macop, reducing the overall efficiency. As shown in Fig.~\ref{sensitivity main}, we can find that when $n_p \le 4$, the performance of Macop does improve with increasing population size. However, there is no further improvement as we continue to increase $n_p$, proving that $n_p = 4$ is the best setting in scenario PP2. More detailed analysis of other important hyperparameters is provided in App.~\ref{completesensiti}.


\section{Final Remarks}

We propose a novel approach to multi-agent policy learning called Macop, which is designed to enhance the coordination abilities of controllable agents when working with diverse teammates. Our approach starts by framing the problem as an CT-Dec-POMDP. This framework entails training the ego-system with sequentially generated groups of teammates until convergence is achieved. Empirical results obtained across various environments, compared against multiple baseline methods, provide strong evidence of its effectiveness.
Looking ahead, in scenarios where we operate under a few-shot setting and need to collect some trajectories for an optimal head during policy deployment, developing mechanisms such as context-based recognition could be a potential future solution.  Additionally, an intriguing direction for future research involves harnessing the capabilities of large language models~\cite{wang2023survey} like ChatGPT~\cite{liu2023summary} to expedite the learning process and further enhance the generalization capabilities of our approach.

\bibliographystyle{ACM-Reference-Format}
\bibliography{sample-base}

\appendix

\clearpage
\section{Appendix}

\subsection{Related Work} \label{relatedwork}
Many real-world problems can often be effectively modeled as multi-agent systems~\cite{DBLP:journals/access/DorriKJ18}. Harnessing the problem-solving prowess of deep reinforcement learning~\cite{Wang2020DeepRL}, Multi-Agent Reinforcement Learning (MARL)~\cite{zhang2021multi} has achieved significant success across diverse domains. Moreover, when agents share a common goal, the problem falls under the category of cooperative MARL~\cite{oroojlooy2023review}, showing impressive progress in various areas such as path finding~\cite{sartoretti2019primal}, active voltage control~\cite{DBLP:conf/nips/WangXGSG21}, and dynamic algorithm configuration~\cite{xue2022multiagent}. Within the wide range of research domains, building an agent capable of cooperating and coordinating with different or even previously unknown teammates remains a fundamental challenge~\cite{dafoe2021cooperative}. Recent approaches, such as ad-hoc teamwork (AHT)~\cite{mirsky2022survey}, zero-shot coordination (ZSC)~\cite{treutlein2021new}, and few-shot teamwork (FST)~\cite{fosong2022few}, have been developed to address this challenge. AHT involves designing agents that can effectively collaborate with new teammates without prior coordination~\cite{stone2010ad}, including aspects like teammate type inference~\cite{barrett2015cooperating,chen2020aateam}, changing point detection~\cite{Ravula2019AdHT}, partial observation solving~\cite{gu2021online}, and adversarial training~\cite{fujimoto2022ad}. ZSC addresses the problem of independently training two or more agents in a cooperative game, ensuring that their strategies are compatible and achieve high returns when paired together at test time~\cite{treutlein2021new}. This line of work includes diversity measurement~\cite{lupu2021trajectory,zhao2023maximum}, training paradigm design~\cite{hu2020other,strouse2021collaborating,xue2022heterogeneous}, investigation of human bias~\cite{yu2023learning}, generation of diverse teammates~\cite{lipo}, and policy co-evolution for heterogeneous settings~\cite{xue2022heterogeneous}. Furthermore, in the FST setting~\cite{fosong2022few}, skilled agents trained as a team to complete one task are combined with skilled agents from different tasks, and they must collectively learn to adapt to an unseen but related tasks~\cite{ding2023coordination}.

One factor for high coordination generalization agents is fostering diversity among teammates, one approach is Fictitious Co-Play (FCP)~\cite{strouse2021collaborating}. This method involves training a controllable agent partner to act as the optimal response to a group of self-play agents and their past checkpoints over the course of training. However, FCP lacks an explicit mechanism to enforce diversity among teammates. To address this limitation, other approaches have been proposed to encourage diversity in teammate behavior. For example, TrajeDi~\cite{lupu2021trajectory} introduces an auxiliary loss term that enhances team diversity by evaluating the trajectories generated by different teams. MEP~\cite{zhao2023maximum} presents a maximum entropy population-based training scheme that mitigates distributional shifts when collaborating with previously unencountered partners.
Another approach, LIPO~\cite{lipo}, focuses on cultivating diverse behaviors by assessing policy compatibility, which measures the effectiveness of policies in coordinating actions. Alternatively, MAZE~\cite{xue2022heterogeneous} addresses the challenge of heterogeneous coordination and introduces a coevolution-based method involving the simultaneous evolution of two populations: agents and partners. Despite the merits of these approaches, they often assume that the process of generating teammates is independent from the optimization objective of coordination policies. This assumption limits their ability to encompass a wide range of teammate policies, subsequently undermining their performance in effectively coordinating actions with previously unseen teammates.

Another related topic is continual reinforcement learning~\cite{khetarpal2022towards,abel2023definition}, which has garnered increasing attention in recent years, focusing on enabling agents to learn sequentially across different tasks. Various methods have been proposed to tackle this challenge. Among the previous works, EWC~\cite{kirkpatrick2017overcoming} employs $l_2$-distance-based weight regularization with previously learned weights, necessitating additional supervision information to select a specific Q-function head and task-specific exploration schedules for different tasks. On the other hand, CLEAR~\cite{Rolnick2018ExperienceRF} is a task-agnostic method that does not require task information during the continual learning process. It stores a large experience replay buffer and addresses the forgetting problem by sampling data from previous tasks.
Other approaches, such as HyperCRL~\cite{huang2021continual} and \cite{kessler2022surprising}, utilize learned world models to enhance the efficiency of continual learning. To address scalability concerns in scenarios with a large number of tasks, LLIRL~\cite{DBLP:journals/tnn/WangCD22} decomposes the task space into subsets and employs the Chinese Restaurant Process to expand the neural network, making continual reinforcement learning more efficient.
OWL~\cite{DBLP:conf/aaai/KesslerPBZR22} is a recently proposed approach that uses a multi-head architecture to achieve high learning efficiency. CSP~\cite{gaya2023building}, on the other hand, incrementally builds a subspace of policies for training a reinforcement learning agent on a sequence of tasks.
Regarding the multi-agent continual learning problem, \cite{nekoei2021continuous} investigate whether agents can coordinate with unseen agents by introducing a multi-agent learning testbed based on Hanabi. However, it only considers uni-modal coordination among tasks. In contrast, MACPro~\cite{yuan2023multi} proposes an approach for multi-agent continual coordination via progressive task contextualization. It obtains a factorized policy using shared feature extraction layers but separate independent task heads, each specializing in a specific class of tasks. Nevertheless, MACPro requires a handcrafted  environment for testing efficiency, which is impractical for real-world applications with unpredictable task scenarios.


\subsection{Proofs for Theorems} \label{ProofsforTheorems}
We first show that a reliable proxy effectively measures the dissimilarity between teammates' policies.

\begin{assumption}
    $P(s'|s, \boldsymbol{a})>0, \forall s, \boldsymbol{a}, s'.$
\end{assumption}
\begin{theorem}
    We define the TV divergence between teammates' policies as $ D_{TV}^{\min}(\boldsymbol{\pi}_{\text{tm}}^i||\boldsymbol{\pi}_{\text{tm}}^j)=\min_{\boldsymbol{\tau}_{\text{tm}}}D_{TV}(\boldsymbol{\pi}_{\text{tm}}^i(\cdot|\boldsymbol{\tau}_{\text{tm}})||\boldsymbol{\pi}_{\text{tm}}^j(\cdot|\boldsymbol{\tau}_{\text{tm}}))$, where $D_{TV}(p||q)=\frac{1}{2}\sum_{x}|p(x)-q(x)|\in[0, 1]$ is the total variantion divergence for the discrete probability distribution $p$ and $q$. We assume that $\min_{\boldsymbol{\tau}_{\text{tm}}, \boldsymbol{a}_{\text{tm}}}\boldsymbol{\pi}^x_{\text{tm}}(\boldsymbol{a}_{\text{tm}}|\boldsymbol{\tau}_{\text{tm}})=\delta>0, x\in \{i, j\}$. Then the following inequality holds:
    \begin{equation}
        \begin{aligned}
            &d(\boldsymbol{\pi}_{\text{tm}}^i, \boldsymbol{\pi}_{\text{tm}}^j)> \epsilon, \quad if \\
             D_{TV}^{\min}(\boldsymbol{\pi}_{\text{tm}}^i||\boldsymbol{\pi}_{\text{tm}}^j)&>\frac{k(k-1)\delta}{2} \min\{1-(1-\epsilon)^{\frac{1}{T}}, (1+\epsilon)^{\frac{1}{T}}-1\},\\
             \text{where}\quad k=|\mathcal{A}_{\text{tm}}|.
        \end{aligned}
    \end{equation}
    \label{thm1}
\end{theorem}
\begin{proof}

    To simplify notation, we use $i$, $j$ to represent $\boldsymbol{\pi}_{\text{tm}}^i$ , $\boldsymbol{\pi}_{\text{tm}}^j$, respectively, and omit the subscript "tm" when no ambiguity arises.

    To prove the theorem, we can transform the problem into determine the value of $x$, $D_{TV}^{\min}(i||j)>x$, i.e., $\sum_{a}|\boldsymbol{\pi}^i(a|\tau)-\boldsymbol{\pi}^j(a|{\tau})|>2x$ holds for all possible trajectories ${\tau} $, so that $d(i, j)>\epsilon$ .

    Let $r(a|{\tau})=\frac{\boldsymbol{\pi}^i(a|\tau)}{\boldsymbol{\pi}^j(a|\tau)}$, and $a^A=\arg\max_{a}|\boldsymbol{\pi}^a(a|\tau)-\boldsymbol{\pi}^j(a|\tau)|$. Clearly, $|\boldsymbol{\pi}^i(a^A|\tau)-\boldsymbol{\pi}^j(a^A|\tau)|>\frac{2x}{k}$, where $k=|\mathcal{A}_{\text{tm}}|$. Without loss of generality, we eliminate the absolute value $|\cdot|$ and observe that two situations may arise for any trajectory $\tau$:
    \begin{enumerate}
        \item $\boldsymbol{\pi}^i(a^A|\tau)-\boldsymbol{\pi}^j(a^A|\tau)>\frac{2x}{k}$, which implies $r(a^A|\tau)>1+\frac{2x}{k\delta}$.
        \item $\boldsymbol{\pi}^i(a^A|\tau)-\boldsymbol{\pi}^j(a^A|\tau)<-\frac{2x}{k}$, which implies $r(a^A|\tau)<1-\frac{2x}{k\delta}$.
    \end{enumerate}
    
    To prove that $d(i, j)>\epsilon$, we need to guarantee that $d(i, j)=|1-\prod_{t=0}^{T-1}r(a_t|\tau_t)|>\epsilon$. Equivalently, we seek to prove that $\exists \tau$ such that $\prod_{t=0}^{T-1}r(a_t|\tau_t)<1-\epsilon$ or $\prod_{t=0}^{T-1}r(a_t|\tau_t)>1+\epsilon$. However, relying on the above "or" constraints for each local point $\tau_t$ does not guarantee a global constraint for $\prod_{t=0}^{T-1}r(a_t|\tau_t)$. To address this, we proceed to derive more general constraints.

    Considering the normalization of the policy, we have $\sum_a (\boldsymbol{\pi}^i(a|\tau)-\boldsymbol{\pi}^j(a|\tau))=0$. Under constraint (1), i.e., $\boldsymbol{\pi}^i(a^A|\tau)-\boldsymbol{\pi}^j(a^A|\tau)>\frac{2x}{k}$, it follows that $\sum_{a, a\neq a^A} (\boldsymbol{\pi}^i(a|\tau)-\boldsymbol{\pi}^j(a|\tau))<-\frac{2x}{k}$. Let $a^{B1}=\arg\min_{a}\boldsymbol{\pi}^i(a|\tau)-\boldsymbol{\pi}^j(a|\tau)$, and we obtain $\boldsymbol{\pi}^i(a^{B1}|\tau)-\boldsymbol{\pi}^j(a^{B1}|\tau)<-\frac{2x}{k(k-1)}$, which leads to $r(a^{B1}|\tau)<1-\frac{2x}{k(k-1)\delta}$. Similarity, for constraint (2), let $a^{B2}=\arg\max_{a}\boldsymbol{\pi}^i(a|\tau)-\boldsymbol{\pi}^j(a|\tau)$, we have $r(a^{B2}|\tau)>1+\frac{2x}{k(k-1)\delta}$.

    Combining the above constraints, we will find that for any $\tau$, one of the following inequalities always hold:

    \begin{enumerate}
        \item $r(a^A|\tau)>1+\frac{2x}{k\delta}$ and $r(a^{B1}|\tau)<1-\frac{2x}{k(k-1)\delta}$.
        \item $r(a^A|\tau)<1-\frac{2x}{k\delta}$ and $r(a^{B2}|\tau)>1+\frac{2x}{k(k-1)\delta}$.
    \end{enumerate}

    Based on the assumption, for any trajectory $\tau$ in $d(i, j)$, we can transform the $(\tau_t, a_t, \tau_{t+1})$ into $(\tau_t, a', \tau_{t+1}), \forall a'$ thus reconstructing a possible trajectory. This allows us to assign $r(a|\tau_t)$ either a lower bound or an upper bound for any $t$.

    Taking the lower bound as an example, we solve the following inequality:
    \begin{equation}
        \begin{aligned}
            \prod_{t=0}^{T-1} r(a_t|\tau_t) \geq & (1+\frac{2x}{k\delta})^p (1+\frac{2x}{k(k-1)\delta})^{T-p}\\
            > & (1+\frac{2x}{k(k-1)\delta})^{T}\\
            \geq& 1+\epsilon, 
        \end{aligned}
    \end{equation}
    and derive that $x=\frac{k(k-1)\delta}{2}((1+\epsilon)^{\frac{1}{T}}-1)$. Similarity, by solving the inequality under the upper bound constraint, we obtain $x=\frac{k(k-1)\delta}{2}(1-(1-\epsilon)^{\frac{1}{T}})$. Combining the two solutions, we finally prove that if $D_{TV}^{\min}(\boldsymbol{\pi}_{\text{tm}}^i||\boldsymbol{\pi}_{\text{tm}}^j)>\frac{k(k-1)\delta}{2} \min\{1-(1-\epsilon)^{\frac{1}{T}}, (1+\epsilon)^{\frac{1}{T}}-1\}$, then $d(\boldsymbol{\pi}_{\text{tm}}^i, \boldsymbol{\pi}_{\text{tm}}^j)> \epsilon$.
\end{proof}
The idea behind Thm.~\ref{thm1} is based on the fact that a significant total variation divergence tends to result in dissimilarities between teammates' policies. 
Due to the computational complexity associated with directly calculating the divergence concerning the trajectory distribution, a more practical approach is to maximize the total variation divergence instead.

Next, we highlight the relationship between total variation divergence and Jensen-Shannon divergence (JSD)~\cite{jsd, lin1991divergence}: $\text{JSD}(p||q)$ $\leq D_{TV}(p||q)$, indicating that JSD serves as a lower bound for total variation divergence. The expression for JSD is given by $\text{JSD}(p||q)=\frac{1}{2}(D_{KL}(p||\frac{p+q}{2})+D_{KL}(q||\frac{p+q}{2}))$ and $D_{KL}$ is the Kullback-Leibler (KL) divergence. While the motivation behind this relationship is rooted in theory, it becomes computationally inefficient to compute the maximum JSD between teammates' policies in the population. Consequently, we adopt a heuristic approximation, similar to TRPO~\cite{trpo}, and considers the average, non-pair-wise JSD:
\begin{equation}
    \begin{aligned}
        \mathcal{L}_{\text{div}}=&\mathbb{E}_{s}[\text{JSD}(\{\boldsymbol{\pi}_{\text{tm}}^{i}\}_{i=1}^{n_p})]\\
        =&\mathbb{E}_{s}[\frac{1}{n_p}\sum_{i=1}^{n_p}D_{KL}(\boldsymbol{\pi}_{\text{tm}}^i(\cdot|s)||\bar{\boldsymbol{\pi}}_{\text{tm}}(\cdot|s))],
    \end{aligned}
    \label{loss_div_app}
\end{equation}
where $\bar{\boldsymbol{\pi}}_{\text{tm}}(\cdot|s)=\frac{1}{n_p}\sum_{i=1}^{n_p}\boldsymbol{\pi}_{\text{tm}}^i(\cdot|s)$ is the average policy of the population.

\begin{assumption}
    $0<G(\boldsymbol{\tau})\leq 1, \forall\boldsymbol{\tau}.$
\end{assumption}
\begin{theorem}
    Given the controllable agents $\boldsymbol{\pi}_{\text{ego}}$ and teammate policy $\boldsymbol{\pi}_{\text{tm}}$, $\forall \boldsymbol{\pi}_{\text{tm}}'$,  $\boldsymbol{\pi}_{\text{tm}}, \boldsymbol{\pi}_{\text{tm}}'$ are $\epsilon-$similar policies. Then we have
    $(1-\epsilon)\mathcal{J}(\langle\boldsymbol{\pi}_{\text{ego}},\boldsymbol{\pi}_{\text{tm}})\leq\mathcal{J}(\langle\boldsymbol{\pi}_{\text{ego}}, \boldsymbol{\pi}_{\text{tm}}'\rangle)\leq (1+\epsilon)\mathcal{J}(\langle\boldsymbol{\pi}_{\text{ego}},\boldsymbol{\pi}_{\text{tm}}\rangle)$.
    \label{thm2_app}
\end{theorem}

\begin{proof}
    From the definition of the expected return of a policy, we have
    \begin{equation}
        \begin{aligned}
        \mathcal{J}(\langle\boldsymbol{\pi}_{\text{ego}}, \boldsymbol{\pi}_{\text{tm}}'\rangle)=
        &\mathbb E_{\boldsymbol{\tau}\sim\rho(\langle\boldsymbol{\pi}_{\text{ego}},\boldsymbol{\pi}_{\text{tm}}'\rangle)}[G(\boldsymbol{\tau})]\\
        =&\int_{\boldsymbol{\tau}} P(\boldsymbol{\tau}|\langle\boldsymbol{\pi}_{\text{ego}},\boldsymbol{\pi}_{\text{tm}}'\rangle) G(\boldsymbol{\tau})\\
        =&\int_{\boldsymbol{\tau}} \frac{P(\boldsymbol{\tau}|\langle\boldsymbol{\pi}_{\text{ego}},\boldsymbol{\pi}_{\text{tm}}\rangle)}{P(\boldsymbol{\tau}|\langle\boldsymbol{\pi}_{\text{ego}},\boldsymbol{\pi}_{\text{tm}}\rangle)}  P(\boldsymbol{\tau}|\langle\boldsymbol{\pi}_{\text{ego}},\boldsymbol{\pi}_{\text{tm}}'\rangle)G(\boldsymbol{\tau})\\
        =&\int_{\boldsymbol{\tau}} P(\boldsymbol{\tau}|\langle\boldsymbol{\pi}_{\text{ego}},\boldsymbol{\pi}_{\text{tm}}\rangle)\frac{P(\boldsymbol{\tau}|\langle\boldsymbol{\pi}_{\text{ego}},\boldsymbol{\pi}_{\text{tm}}'\rangle)}{P(\boldsymbol{\tau}|\langle\boldsymbol{\pi}_{\text{ego}},\boldsymbol{\pi}_{\text{tm}}\rangle)} G(\boldsymbol{\tau})\\
        =&\mathbb E_{\boldsymbol{\tau}\sim\rho(\langle\boldsymbol{\pi}_{\text{ego}},\boldsymbol{\pi}_{\text{tm}}\rangle)}[\frac{P(\boldsymbol{\tau}|\langle\boldsymbol{\pi}_{\text{ego}},\boldsymbol{\pi}_{\text{tm}}'\rangle)}{P(\boldsymbol{\tau}|\langle\boldsymbol{\pi}_{\text{ego}},\boldsymbol{\pi}_{\text{tm}}\rangle)}G(\boldsymbol{\tau})].
        \end{aligned}
    \end{equation}
    From the definition of $\epsilon$-similar policies, we know that $1-\epsilon\leq\frac{P(\boldsymbol{\tau}|\langle\boldsymbol{\pi}_{\text{ego}},\boldsymbol{\pi}_{\text{tm}}'\rangle)}{P(\boldsymbol{\tau}|\langle\boldsymbol{\pi}_{\text{ego}},\boldsymbol{\pi}_{\text{tm}}\rangle
    )}\leq 1+\epsilon$, which implies that
    \begin{equation}
        \begin{aligned}
        (1-\epsilon)G(\boldsymbol{\tau})\leq\frac{P(\boldsymbol{\tau}|\langle\boldsymbol{\pi}_{\text{ego}},\boldsymbol{\pi}_{\text{tm}}'\rangle)}{P(\boldsymbol{\tau}|\langle\boldsymbol{\pi}_{\text{ego}},\boldsymbol{\pi}_{\text{tm}}\rangle)}G(\boldsymbol{\tau})\leq (1+\epsilon)G(\boldsymbol{\tau}).\\
        \end{aligned}
    \end{equation}

    Further more, we observe the lower and upper bound for the expected return $\mathcal{J}(\langle\boldsymbol{\pi}_{\text{ego}}, \boldsymbol{\pi}_{\text{tm}}'\rangle)$:
    \begin{equation}
        \begin{aligned}
        \mathcal{J}(\langle\boldsymbol{\pi}_{\text{ego}}, \boldsymbol{\pi}_{\text{tm}}'\rangle)\geq& 
        (1-\epsilon)\mathbb E_{\boldsymbol{\tau}\sim\rho(\langle\boldsymbol{\pi}_{\text{ego}},\boldsymbol{\pi}_{\text{tm}}\rangle)}[G(\boldsymbol{\tau})]\\
        \mathcal{J}(\langle\boldsymbol{\pi}_{\text{ego}}, \boldsymbol{\pi}_{\text{tm}}'\rangle)\leq& (1+\epsilon)\mathbb E_{\boldsymbol{\tau}\sim\rho(\langle\boldsymbol{\pi}_{\text{ego}},\boldsymbol{\pi}_{\text{tm}}\rangle)}[G(\boldsymbol{\tau})].\\
        \end{aligned}
    \end{equation}
    By replacing $\mathbb E_{\boldsymbol{\tau}\sim\rho(\langle\boldsymbol{\pi}_{\text{ego}},\boldsymbol{\pi}_{\text{tm}}\rangle)}[G(\boldsymbol{\tau})]$ with $\mathcal{J}(\langle\boldsymbol{\pi}_{\text{ego}}, \boldsymbol{\pi}_{\text{tm}}\rangle)$, we finally derive that $(1-\epsilon)\mathcal{J}(\langle\boldsymbol{\pi}_{\text{ego}},\boldsymbol{\pi}_{\text{tm}})\leq\mathcal{J}(\langle\boldsymbol{\pi}_{\text{ego}}, \boldsymbol{\pi}_{\text{tm}}'\rangle)\leq (1+\epsilon)\mathcal{J}(\langle\boldsymbol{\pi}_{\text{ego}},\boldsymbol{\pi}_{\text{tm}}\rangle)$.
\end{proof}

\subsection{More Details of Macop}
\label{macopdetails}
\subsubsection{The Overall Workflow of Macop}

We introduce the pseudo-codes for both the training and testing phases of Macop in this part, in Alg.~\ref{algo1} and Alg.~\ref{algo2}.

~\\

\begin{breakablealgorithm}
\caption{Macop: Training Phase}\label{algo1}
\begin{algorithmic}[1]

    \Statex \textbf{Input:} controllable agents $\boldsymbol{\pi}_{\text{ego}}$, population size $n_p$, minimum iteration $N_{\min}$, maximum iteration $N_{\max}$, stopping threshold $\xi$, head expansion threshold $\lambda$.
    \State Initialize and train teammate population $\mathcal{P}_{\text{tm}}^0$ with $\mathcal{L}_{\text{tm}}(\alpha_{\text{incom}} = 0)$
    \For{$iter=1$ to $N_{\max}$}
        \If{$iter = 1$}
            \State $\mathcal{P}_{\text{tm}}^{iter}\leftarrow \mathcal{P}_{\text{tm}}^{iter-1}$
        \Else
            \State $\mathcal{P}_{\text{tm}}^{\text{parent}}\leftarrow \mathcal{P}_{\text{tm}}^{iter-1}$
            
            \State $\mathcal{P}_{\text{tm}}^{\text{offspring}}\leftarrow \text{Mutation}(\boldsymbol{\pi}_{\text{ego}}, \mathcal{P}_{\text{tm}}^{iter-1})$
            \State $\mathcal{P}_{\text{tm}}^{iter}\leftarrow \text{Selection}(\mathcal{P}_{\text{tm}}^{\text{parent}}, \mathcal{P}_{\text{tm}}^{\text{offspring}})$
        \EndIf
        \If{$iter \ge N_{\min}$}
            \State $C\leftarrow \frac{\min_{i}\mathcal{J}(\langle\boldsymbol{\pi}_{\text{ego}},\boldsymbol{\pi}_{\text{tm}}^i \rangle)}{\frac{1}{n_p}\sum_{i=1}^{n_p} \mathcal{J}_{\text{sp}}( \boldsymbol{\pi}_{\text{tm}}^i)},~\boldsymbol{\pi}_{\text{tm}}^i \in \mathcal{P}_{\text{tm}}$
            \If{$C\geq\xi$}
                \State break \Comment{terminate}
            \EndIf
        \EndIf
        \For{$\boldsymbol{\pi}_{\text{tm}}\in\mathcal{P}_{\text{tm}}^{iter}$}
            \State $\boldsymbol{\pi}_{\text{ego}}\leftarrow\text{ContinualLearning}(\boldsymbol{\pi}_{\text{ego}}, \boldsymbol{\pi}_{\text{tm}}, \lambda)$
        \EndFor
    \EndFor
    \Return $\boldsymbol{\pi}_{\text{ego}}$

\end{algorithmic}
\end{breakablealgorithm}

\begin{breakablealgorithm}
\caption{Mutation}
\begin{algorithmic}[1]

    \Statex \textbf{Input:} controllable agents $\boldsymbol{\pi}_{\text{ego}}$, teammate population $\mathcal{P}_{\text{tm}}$.
    \State Initialize cross-play buffer $\{\mathcal{D}_{\text{tm}}^{\text{XP}, i}\}_{i=1}^{n_p}$, self-play buffer $\{\mathcal{D}_{\text{tm}}^{\text{SP}, i}\}_{i=1}^{n_p}$, complementary polices $\{\bar{\boldsymbol{\pi}}_{\text{ego}}^i\}_{i=1}^{n_p}$
    \State $t\leftarrow 0$
    \While{$t\leq t_{\text{tm}}$}
        \For{$\boldsymbol{\pi}_{\text{tm}}^i\in \mathcal{P}_{\text{tm}}$}
            \State $\boldsymbol{\tau}^{\text{XP}}\leftarrow env.rollout( \boldsymbol{\pi}_{\text{ego}}, \boldsymbol{\pi}_{\text{tm}}^i)$\Comment{xp trajectory}
            \State $t\leftarrow t+\boldsymbol{\tau}^{\text{XP}}.length$
            \State $\mathcal{D}_{\text{tm}}^{\text{XP}, i}\leftarrow\mathcal{D}_{\text{tm}}^{\text{XP}, i}\cup\{\boldsymbol{\tau}^{\text{XP}}\}$
            \State $\boldsymbol{\tau}^{\text{SP}}\leftarrow env.rollout(\bar{\boldsymbol{\pi}}_{\text{ego}}^i, \boldsymbol{\pi}_{\text{tm}}^i)$\Comment{sp trajectory}
            \State $t\leftarrow t+\boldsymbol{\tau}^{\text{SP}}.length$
            \State $\mathcal{D}_{\text{tm}}^{\text{SP}, i}\leftarrow\mathcal{D}_{\text{tm}}^{\text{SP}, i}\cup\{\boldsymbol{\tau}^{\text{SP}}\}$
        \EndFor
        \State Optimize $\mathcal{P}_{\text{tm}}$ via $\mathcal{L}_{\text{tm}}$ in Eqn.~\ref{loss_tm} \Comment{optimization}
    \EndWhile
    \State \Return $\mathcal{P}_{\text{tm}}$

\end{algorithmic}
\end{breakablealgorithm}

\begin{breakablealgorithm}
\caption{Selection}
\begin{algorithmic}[1]

    \Statex \textbf{Input:} parent population $\mathcal{P}_{\text{tm}}^{\text{parent}}$, offspring population $\mathcal{P}_{\text{tm}}^{\text{offspring}}$.
    \State $\mathcal{P}_{\text{tm}}\leftarrow \mathcal{P}_{\text{tm}}^{\text{parent}} \cup \mathcal{P}_{\text{tm}}^{\text{offspring}}$
    \State $\mathcal{P}_{1}\leftarrow\text{AscendSelect}(\mathcal{P}_{\text{tm}}, \text{``self-play return''}, \lfloor\frac{n_p}{2}\rfloor)$
    \State $\mathcal{P}_{\text{tm}}\leftarrow\mathcal{P}_{\text{tm}}\backslash\mathcal{P}_1$
    \State $\mathcal{P}_{2}\leftarrow\text{DescendSelect}(\mathcal{P}_{\text{tm}}, \text{``cross-play return''}, \lceil\frac{n_p}{2}\rceil)$
    \State $\mathcal{P}_{\text{tm}}\leftarrow\mathcal{P}_{\text{tm}}\backslash\mathcal{P}_2$
    \State \Return $\mathcal{P}_{\text{tm}}$

\end{algorithmic}
\end{breakablealgorithm}

\begin{breakablealgorithm}
\caption{Continual Learning}
\begin{algorithmic}[1]

    \Statex \textbf{Input:} controllable agents $\boldsymbol{\pi}_{\text{ego}}$, teammate group $\boldsymbol{\pi}_{\text{tm}}$, head expansion threshold $\lambda$.
    \State $f, \{h_i\}_{i=1}^m\leftarrow Split(\boldsymbol{\pi}_{\text{ego}})$\Comment{get backbone and heads}
    \State Initialize cross-play replay buffer $\mathcal{D}_{\text{ego}}^{\text{XP}}$, self-play replay buffer $\mathcal{D}_{\text{ego}}^{\text{SP}}$, complementary policy $\bar{\boldsymbol{\pi}}_{\text{tm}}$, new head $h^{\text{new}}$
    \State Compose $\boldsymbol{\pi}_{\text{ego}}^{\text{new}}$ with $f$ and $h^{\text{new}}$

    \State $t\leftarrow 0$
    \While{$t\leq t_{\text{ego}}$}
        \State $\boldsymbol{\tau}^{\text{XP}}\leftarrow env.rollout( \boldsymbol{\pi}_{\text{ego}}, \boldsymbol{\pi}_{\text{tm}})$\Comment{xp trajectory}
        \State $t\leftarrow t+\boldsymbol{\tau}^{\text{XP}}.length$
        \State $\mathcal{D}_{\text{ego}}^{\text{XP}}\leftarrow\mathcal{D}_{\text{ego}}^{\text{XP}}\cup\{\boldsymbol{\tau}^{\text{XP}}\}$
        \State $\boldsymbol{\tau}^{\text{SP}}\leftarrow env.rollout({\boldsymbol{\pi}}_{\text{ego}}, \bar{\boldsymbol{\pi}}_{\text{tm}})$\Comment{sp trajectory}
        \State $t\leftarrow t+\boldsymbol{\tau}^{\text{SP}}.length$
        \State $\mathcal{D}_{\text{tm}}^{\text{SP}, i}\leftarrow\mathcal{D}_{\text{tm}}^{\text{SP}, i}\cup\{\boldsymbol{\tau}^{\text{SP}}\}$
        \State Optimize $\boldsymbol{\pi}_{\text{ego}}^{\text{new}}$ via $\mathcal{L}_{\text{ego}}$ in Eqn.~\ref{loss_ego}
    \EndWhile
    \State Calculate empirical return $\hat R_{new}$ with $\langle\boldsymbol{\pi}_{\text{ego}}^{\text{new}}, \boldsymbol{\pi}_{\text{tm}}\rangle$
    \For{$h_i\in\{h_i\}_{i=1}^m$}
        \State Compose $\boldsymbol{\pi}_{\text{ego}}^i$ with $f$ and $h_i$
        \State Calculate empirical return $\hat R_i$ with $\langle\boldsymbol{\pi}_{\text{ego}}^i, \boldsymbol{\pi}_{\text{tm}}\rangle$
    \EndFor
    \If{$\frac{\hat R_{new}-\max_{i} \{\hat R_i\}_{i=1}^m}{\max_{i} \{\hat R_i\}_{i=1}^m}\geq \lambda$}
        \State Add new head $h^{\text{new}}$ to the policy $\boldsymbol{\pi}_{\text{ego}}$ \Comment{expand}
    \EndIf
    \State \Return $\boldsymbol{\pi}_{\text{ego}}$    
    
\end{algorithmic}
\end{breakablealgorithm}

\begin{breakablealgorithm}
    \caption{Macop: Testing Phase}\label{algo2}
\begin{algorithmic}[1]
    \Statex \textbf{Input:} controllable agents $\boldsymbol{\pi}_{\text{ego}}$, unknown teammate group $\boldsymbol{\pi}_{\text{tm}}$, number of samples $n_{meta}$.
    \State $f, \{h_i\}_{i=1}^m\leftarrow Split(\boldsymbol{\pi}_{\text{ego}})$\Comment{get backbone and heads}
    \State Initialize counter $\{n_i\}_{i=1}^m\leftarrow \{0\}_{i=1}^m$
    \State Initialize $\{\hat R_i\}_{i=1}^m\leftarrow\{0\}_{i=1}^m$
    \While{$Sum(\{n_i\}_{i=1}^m)\leq n_{meta}$}
        \State Choose one head $h_i$
        \State Compose $\boldsymbol{\pi}_{\text{ego}}^i$ with $f$ and $h_i$
        \State Calculate one episode return $R_i$ from $\langle\boldsymbol{\pi}_{\text{ego}}^i, \boldsymbol{\pi}_{\text{tm}}\rangle$
        \State $\hat R_i \leftarrow\hat R_i +\frac{R_i-\hat R_i}{n_i+1}$
        \State $n_i\leftarrow n_i+1$
    \EndWhile
    \State Determine the head index $i^*=\arg\max_{i^*}\{\hat R_i\}_{i=1}^m$
    \State Compose $\boldsymbol{\pi}_{\text{ego}}^{\text{eval}}$ with $f$ and $h_{i^*}$
    \State Evaluate with $\langle\boldsymbol{\pi}_{\text{ego}}^{\text{eval}}, \boldsymbol{\pi}_{\text{tm}}\rangle$
\end{algorithmic}
\end{breakablealgorithm}

\subsubsection{The Architecture, Infrastructure, and Hyperparameters Choices of Macop}

We implement Macop based on the PyMARL\footnote{\url{https://github.com/oxwhirl/pymarl}}~\cite{pymarl} codebase. For agent network architecture, we apply the technique of parameter sharing, so self-play and cross-play can be easily implemented by inputting different agent ids into the agent network. And we design the feature extraction backbone $f_{\phi}$ as a 2-layer MLP and a GRU~\cite{DBLP:conf/emnlp/ChoMGBBSB14}, and the policy head $h_{\psi_i}$ as a 2-layer MLP. The hidden dimension is 64 for both MLP and GRU. The head MLP takes the output of the backbone as input and outputs the Q-value of all actions. The individual Q-values of each agent are then fed into the mixing network to calculate the joint Q-value, according to the existing MARL methods. We select VDN~\cite{vdn} for environment LBF, PP, CN, and QMIX~\cite{qmix} for environment SMAC. We adopt Adam~\cite{DBLP:journals/corr/KingmaB14} as the optimizer with learning rate $5 \times 10^{-4}$. The whole framework is trained end-to-end with collected episodic data on NVIDIA GeForce RTX 2080 Ti and 3090 GPUs with a time cost of about 10 hours in LBF, PP, CN, SMAC1 scenarios, and about 30 hours in SMAC2 scenario.

We use the default hyperparameter settings of PyMARL, e.g., the batch size of trajectories used to calculate the temporal difference error is set to the default value 32. The selection of the additional hyperparameters introduced in our approach, e.g., the size of each teammate population, is listed in Tab.~\ref{hyperparameters}.

\begin{table*}
\small
\setlength{\tabcolsep}{6pt}
\renewcommand{\arraystretch}{1.3}
    \centering
    \caption{Hyperparameters in experiment.}
    \resizebox{0.99\textwidth}{!}{\begin{tabular}{c|c}
        \toprule
        Hyperparameter  & Value \\
        \midrule
        $n_p$ (Population size) & 4 (LBF, PP, CN), 2 (SMAC) \\
        $\alpha_{\text{incom}}$ (Coefficient of $\mathcal{L}_{\text{incom}}$) & 0.1 \\
        $\alpha_{\text{div}}$ (Coefficient of $\mathcal{L}_{\text{div}}$) & 0.1 \\
        $\alpha_{\text{reg}}$ (Coefficient of $\mathcal{L}_{\text{reg}}$) & 10 \\
        Number of testing episodes & 32 \\
        $N_{\text{min}}$ (Minimum number of iterations) & 4 (LBF, PP, CN), 3 (SMAC) \\
        $N_{\text{max}}$ (Minimum number of iterations) & 10 \\
        $\lambda$ (Threshold value for head expansion) & 0 \\
        $\xi$ (Threshold value for Macop's termination) & 0.5 \\
        $t_{\text{tm}}$ (Timesteps for training teammate population) & 500k (LBF, PP, CN), 600k (SMAC1), 1M (SMAC2) \\
        $t_{\text{ego}}$ (Timesteps for training controllable agents with one teammate group) & 125k (LBF, CN), 150k (PP), 300k (SMAC1), 500k (SMAC2) \\
        \bottomrule
    \end{tabular}}
    \label{hyperparameters}
\end{table*}

\begin{figure*}
\setlength{\abovecaptionskip}{0cm}
  \centering
    \subfigure[Level-based Foraging (LBF1)]{
    \label{env lbf1}
    \includegraphics[width=0.23\textwidth]{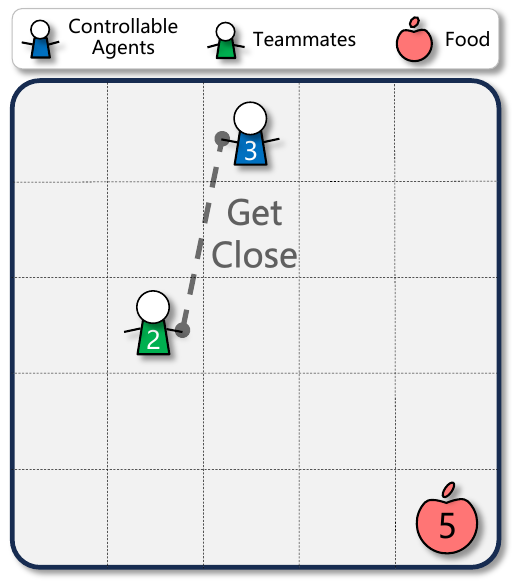}
    }
    \subfigure[Predator Prey (PP1)]{
    \label{env pp1}
    \includegraphics[width=0.23\textwidth]{Figures/envs/pp1.pdf}
    }
    \subfigure[Cooperative Navigation (CN2)]{
    \label{env cn2}
    \includegraphics[width=0.23\textwidth]{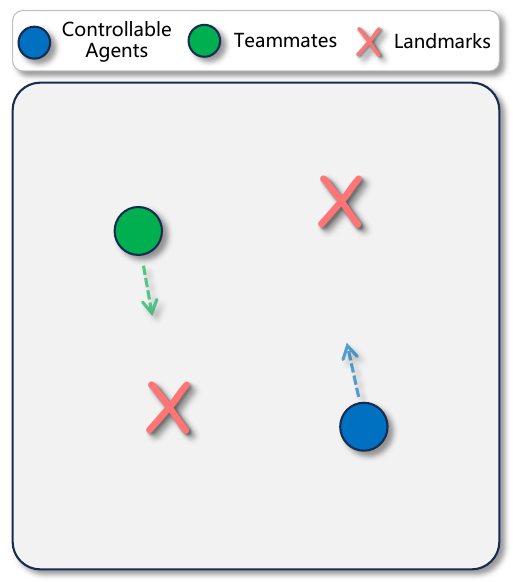}
    }
    \subfigure[{\scriptsize Starcraft Multi-agent Challenge} (SMAC1)]{
    \label{env smac1}
    \includegraphics[width=0.23\textwidth]{Figures/envs/smac1.pdf}
    }
    \subfigure[Level-based Foraging (LBF4)]{
    \label{env lbf4}
    \includegraphics[width=0.23\textwidth]{Figures/envs/lbf4.pdf}
    }
    \subfigure[Predator Prey (PP2)]{
    \label{env pp2}
    \includegraphics[width=0.23\textwidth]{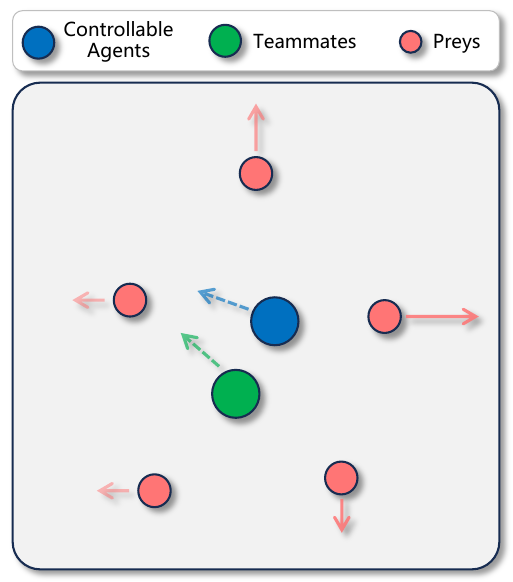}
    }
    \subfigure[Cooperative Navigation (CN3)]{
    \label{env cn3}
    \includegraphics[width=0.23\textwidth]{Figures/envs/cn3.pdf}
    }
    \subfigure[{\scriptsize Starcraft Multi-agent Challenge} (SMAC2)]{
    \label{env smac2}
    \includegraphics[width=0.23\textwidth]{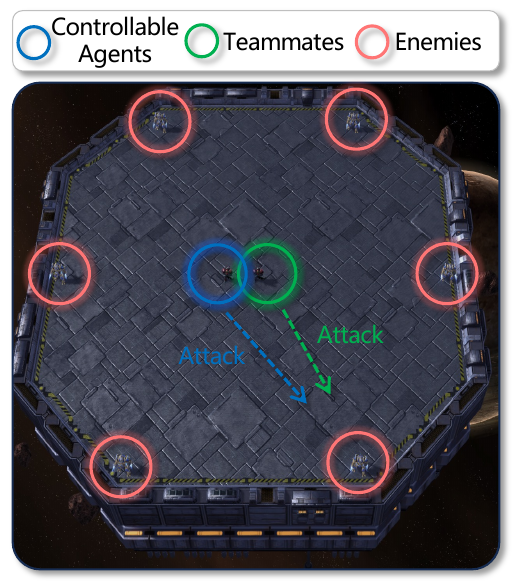}
    }
  \caption{Environments used in this paper.}
  \label{environments_appendix}
\end{figure*}

\subsection{More Details of Experiments}
\label{moreexperiments}
\subsubsection{Environments}\label{appx_env}

We select four multi-agent coordination environments and design two scenarios each as evaluation benchmarks. Four scenarios (LBF4, PP1, CN3, SMAC1) displayed in the manuscript, together with the other four, are shown in Fig.~\ref{environments_appendix}. Here we provide details of all eight scenarios.

Level-based Foraging (LBF)~\cite{lbf} is a discrete grid world game where two agents with varying levels navigate through the grid to collect foods with different levels. Each agent moves one cell at a time in one of the four directions \{up, left, down, right\}. Agents gain reward 1 when they are at a distance of one cell from the food and the sum of their levels matches or exceeds the level of the food.
LBF1 is a $5 \times 5$ grid world with one food at cell $(4, 4)$. Two agents are randomly spawned at cells $\{(0, 0), (0, 1), (1, 0), (1, 1)\}$ at the beginning of an episode. When both agents collect the food, they can gain reward 1, only if they have been at a distance of one cell two times before, or 0 otherwise. Therefore, this scenario requires the agents to not only simply navigate to the food, but also observe the teammate's preferred path and get close to it.
Another scenario named LBF4 is a $6 \times 6$ grid world with four foods at cells $(0, 0), (0, 5), (5, 0), (5, 5)$, and two agents are randomly spawned at cells $\{(2, 2), (2, 3), (3, 2), (3, 3)\}$, requiring the controllable agent to observe the teammate's preferred food and collect it together.

Predator Prey (PP) and Cooperative Navigation (CN) are two benchmarks coming from the popular MPE environment~\cite{maddpg}, where agents and landmarks are represented by circles of different radii in a two-dimension plane. Agents can accelerate in one of the four directions \{up, left, down, right\} and move on the plane to finish the tasks.
In environment PP~\cite{maddpg}, two agents (predators) must together pursue and collide the moving adversaries (preys) at the same time to gain reward 1. Here the preys are controlled by some fixed policies. In PP1, three preys are controlled by a random policy, and in PP2, five preys are controlled by a heuristic policy to run away from the nearest predator. Since multiple preys are spawned in different areas, the predators need to highly coordinate to pursue the same prey. Also, the controllable agent needs to adapt to the teammate's pursuit strategy to successfully block the prey from running.
In environment CN~\cite{maddpg}, $n$ agents navigate to $n$ landmarks separately and receive a reward when all $n$ landmarks have one agent nearby. In CN2, $n=2$, while in CN3, $n=3$ and we have 2 controllable agents, who need to coordinate with the teammates and separately move to different landmarks.

We also conduct experiments in the widely used StarCraft II combat scenario, SMAC~\cite{pymarl}, which involves unit micromanagement tasks. In this setting, ally units (agents) are trained to beat enemy units controlled by the built-in AI. Agents receive a positive reward signal by causing damage to enemies, killing enemies, and winning the battle. On the contrary, agents receive a negative reward signal when they receive damage from enemies, get killed, and lose the battle. 
In SMAC1, 2 ally units called ``marine'' engage 2 enemy groups in battle. These 2 enemy groups are separated on the map, and each enemy group also consists of 2 marines. To gain a higher reward, agents must move together to encounter one enemy group, and adapt to each other's combat strategies to strike the enemies.
In SMAC2, enemies are another type of unit known as ``zealot'' which is hard to be defeated but can only attack nearby units. Six of them are separated on the map. Agents need to attack only one zealot at a time without attracting other zealots' attention. It requires the agents to adapt to each other's combat route to attack the zealots one by one. In conclusion, all scenarios are involved multi-modality and varying behaviors among different teammate groups, requiring the controllable agents' strong coordination ability to finish the task with different diverse teammates.

\subsubsection{Baselines}

\textbf{FCP}~\cite{strouse2021collaborating} first trains a population of teammate policies independently until convergence and then trains the controllable agents by pairing them with these fixed teammate policies in the population. The diversity is induced solely by network random initialization. In our implementation, we set the population size as 6 and setup a replay buffer for each teammate group. Each time, we uniformly sample one teammate group to collect one trajectory, store it into its replay buffer, and sample a batch of trajectories for training. The overall population is trained for 1M timesteps in total. Then, we train the controllable agents for another 1M timesteps. Similarly, we uniformly sample one teammate group to pair with the controllable agents to collect trajectories into the replay buffer for training. \textbf{TrajeDi}~\cite{lupu2021trajectory} applies an auxiliary loss term $\mathcal{L}_{\text{sp}}$ when training the teammate group, while \textbf{LIPO}~\cite{lipo} replaces this term with $\mathcal{J}_{\text{LIPO}} = -\sum_{i\neq j}\mathcal{J}(\langle  \boldsymbol{\pi}_{\text{tm}}^i, \boldsymbol{\pi}_{\text{tm}}^j \rangle)$, as mentioned in the manuscript. The rest implementation details of TrajeDi and LIPO remain the same as FCP. \textbf{EWC}~\cite{kirkpatrick2017overcoming} is a popular regularization-based continual learning approach which maintains knowledge on learned tasks (teammate groups) by selectively slowing down learning on the weights which are important. In our implementation, when learning to coordinate with the $k^{\text{th}}$ teammate group, the loss term is $\mathcal{L}_{\text{ego}} = \mathcal{L}_{\text{com}} + \frac{\mu}{2} \sum\limits_j  F_j (\theta_j - \theta_{k-1, j})^2$, where $F_i$ is the $i^{\text{th}}$ diagonal element of the Fisher information matrix $F$, calculated with temporal difference error. $\theta_{k-1}$ is the saved snapshot of agent network parameters after training with the $(k-1)^{\text{th}}$ teammate group, and $j$ labels each parameter. $\mu$ is an adjustable coefficient to control the trade-off between the current and previous teammate groups. Another replay-based continual learning method, \textbf{CLEAR}~\cite{Rolnick2018ExperienceRF}, stores some trajectories of previously trained teammate groups to rehearse the controllable agents when training with the current teammate group. We split the replay buffer evenly for each trained teammate group in our implementation.
     \begin{figure}[H]
  \centering
    \label{loss}
    \includegraphics[width=0.41\textwidth]{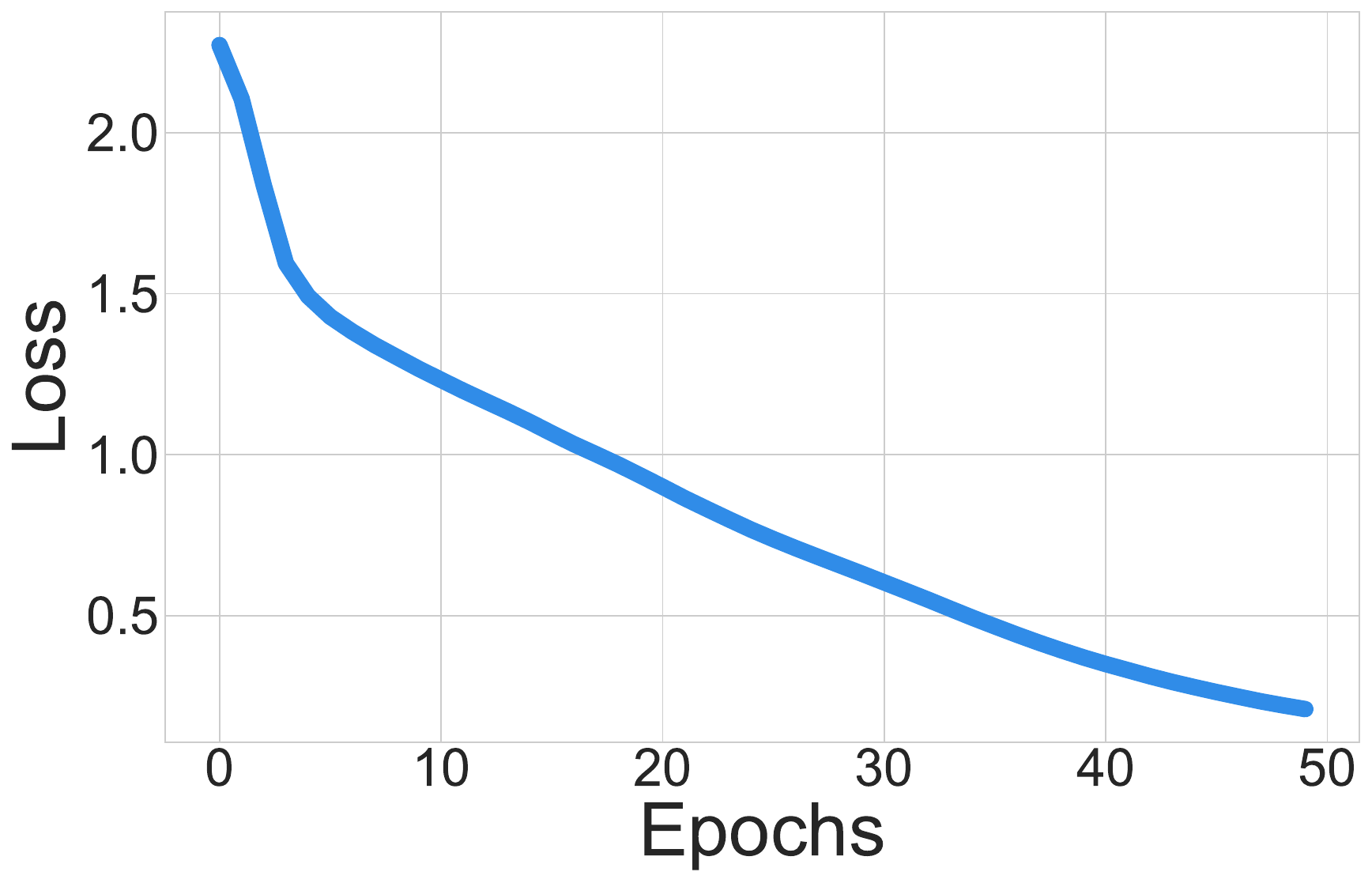}
  \caption{Negative log-likelihood loss for action prediction.}
  \label{loss}
\end{figure}   
\subsubsection{Trajectory Encoder}\label{moreencoder}

In Sec.~\ref{overall algorithm} of manuscript, we utilize a transformer-based trajectory encoder used in MACPro~\cite{DBLP:journals/corr/abs-2305-05116} to derive the feature of self-play trajectories rolled out by different teammate groups. For a trajectory $((s_1, \pmb{a}_1), \cdots, (s_T, \pmb{a}_T))$ with horizon $T$, we input it into a transformer encoder to derive $T$ 10-dimensional embeddings, then we apply a mean-pooling operator to derive the trajectory feature. To train this trajectory encoder, we also utilize a GRU-based~\cite{DBLP:conf/emnlp/ChoMGBBSB14} decoder. For a given trajectory, it takes the trajectory feature and the current state, and predicts the taken joint action. For accurate prediction, the encoder must encode the important feature of the trajectory. Both trajectory encoder and decoder are trained in an end-to-end manner for 50 epochs, with a batch size of 32 and an Adam~\cite{DBLP:journals/corr/KingmaB14} optimizer with learning rate $1 \times 10^{-4}$. The negative log-likelihood loss for action prediction is shown in Fig.~\ref{loss}.

\subsection{The Complete Sensitivity Analysis}
\label{completesensiti}

We provide a detailed sensitivity analysis of multiple hyperparameters in this section. First, two important hyperparameters $\alpha_{\text{incom}}, \alpha_{\text{div}}$ control the teammate generation process. If they are too small, diverse teammate groups cannot be generated efficiently. On the contrary, setting them to a large value will impair the learning of teammate policy. We find that $\alpha_{\text{incom}} = 0.1$ in PP2 and $\alpha_{\text{div}} = 0.1$ in LBF4 perform the best. Another hyperparameter $\alpha_{\text{reg}}$ controls the extent of the regularization on the backbone. A small value fails to prevent catastrophic forgetting, while too strong regularization with large $\alpha_{\text{reg}}$ will constrain the forward transfer with new teammate groups. According to Fig.~\ref{alpha reg}, $\alpha_{\text{reg}} = 10$ is an appropriate value in scenario PP2. The population size $n_p$ is also a critical hyperparameter (Fig.~\ref{population size app}), as analyzed in the manuscript.
For the stopping threshold $\xi$ and the head expansion threshold $\lambda$, we conduct experiments on LBF4 to investigate their sensitivity. $\xi$ determines the acceptance of Macop's convergence, where a small $\xi$ makes Macop stops easily without generating enough incompatible teammates, but a large $\xi$ is hard to converge and the controllable agents might overfit to some specific generated teammate groups. We find that $\xi = 0.5$ performs the best according to Fig.~\ref{stop threshold}. Finally, the head expansion threshold $\lambda$ determines whether a newly trained head should be saved. A small $\lambda$ leads to redundant output heads, increasing the storage overhead as well as the difficulty of selecting the optimal head during evaluation. When the value of $\lambda$ is too large, Macop will fail to be equipped with enough heads to handle the multi-modality of teammate behaviors in complex scenarios. As shown in Fig.~\ref{head threshold}, $\lambda = 0$ is the most appropriate head expansion threshold value.

\begin{figure}
  \centering
    \subfigure[]{
    \label{alpha xp}
    \includegraphics[width=0.26\textwidth]{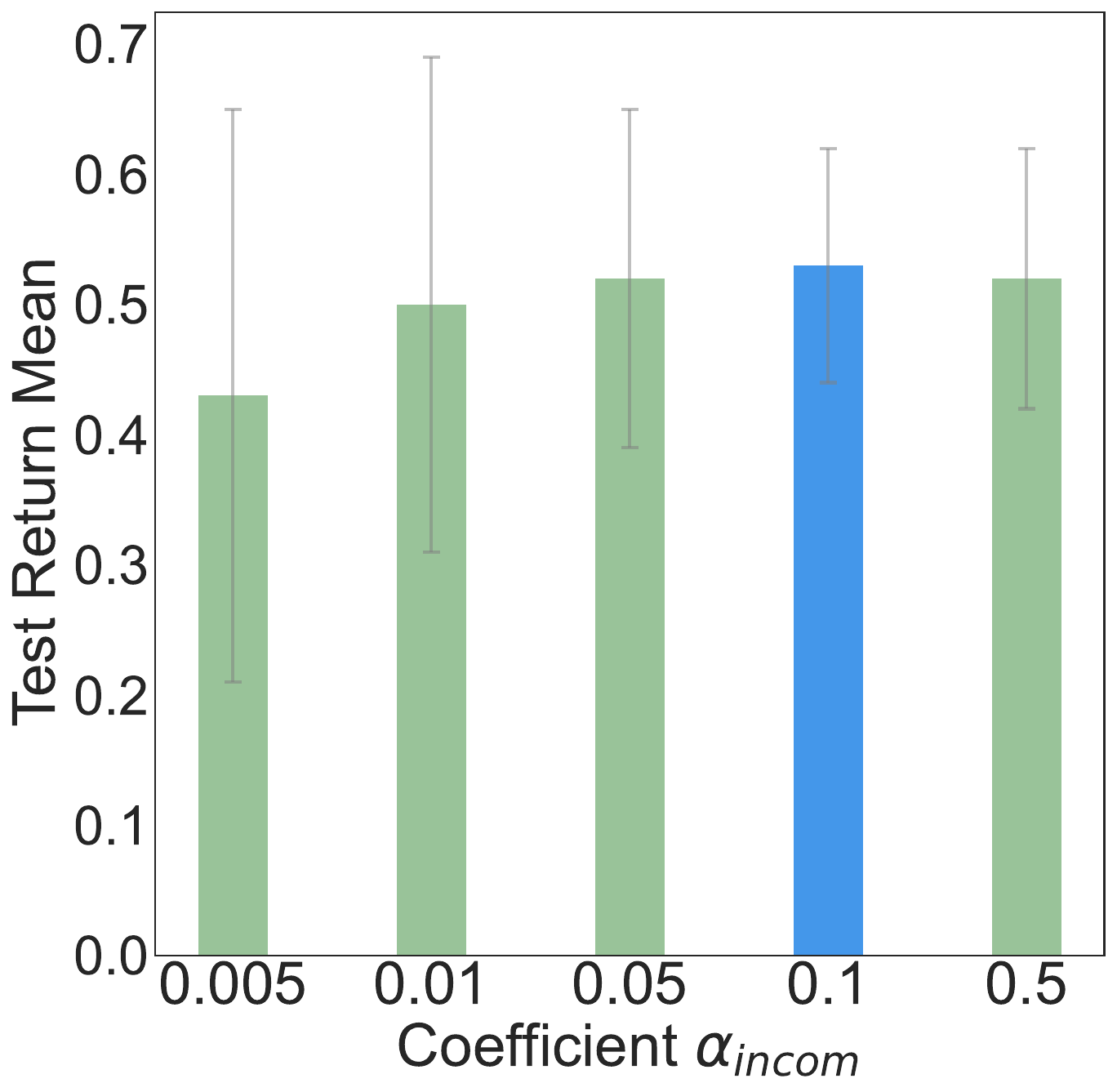}
    }
    \subfigure[]{
    \label{alpha div}
    \includegraphics[width=0.26\textwidth]{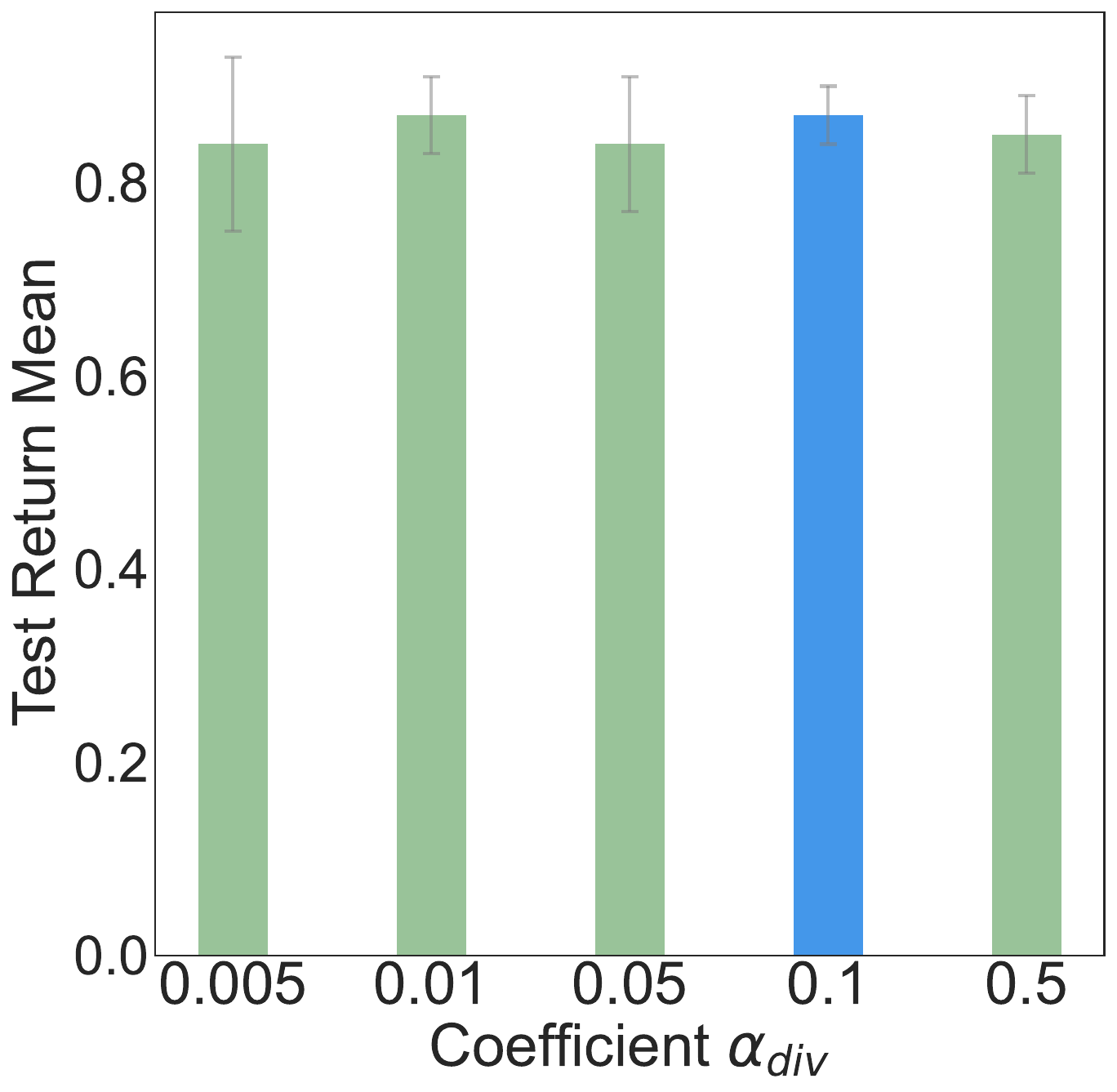}
    }
    \subfigure[]{
    \label{alpha reg}
    \includegraphics[width=0.26\textwidth]{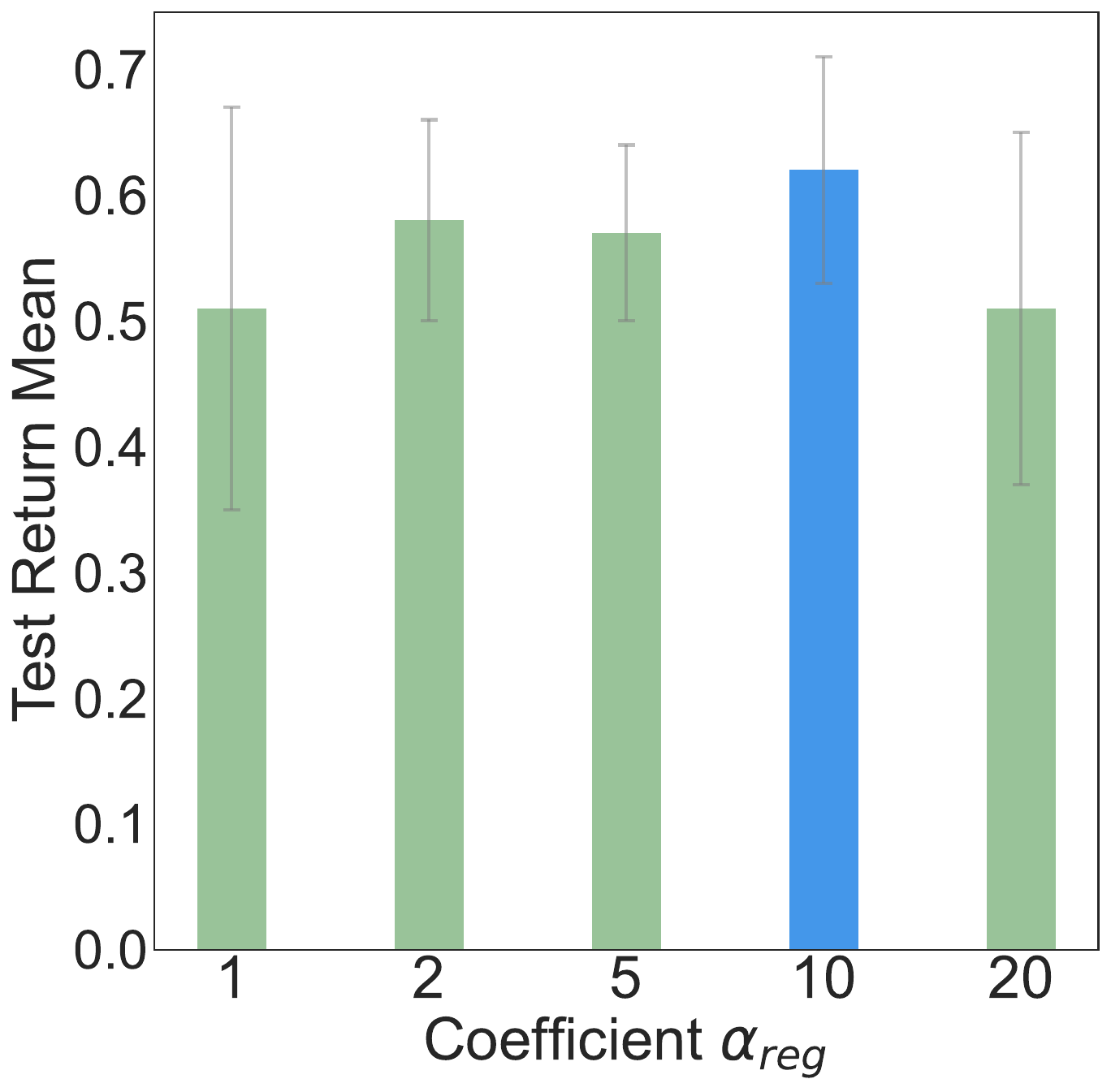}
    }
    \subfigure[]{
    \label{population size app}
    \includegraphics[width=0.26\textwidth]{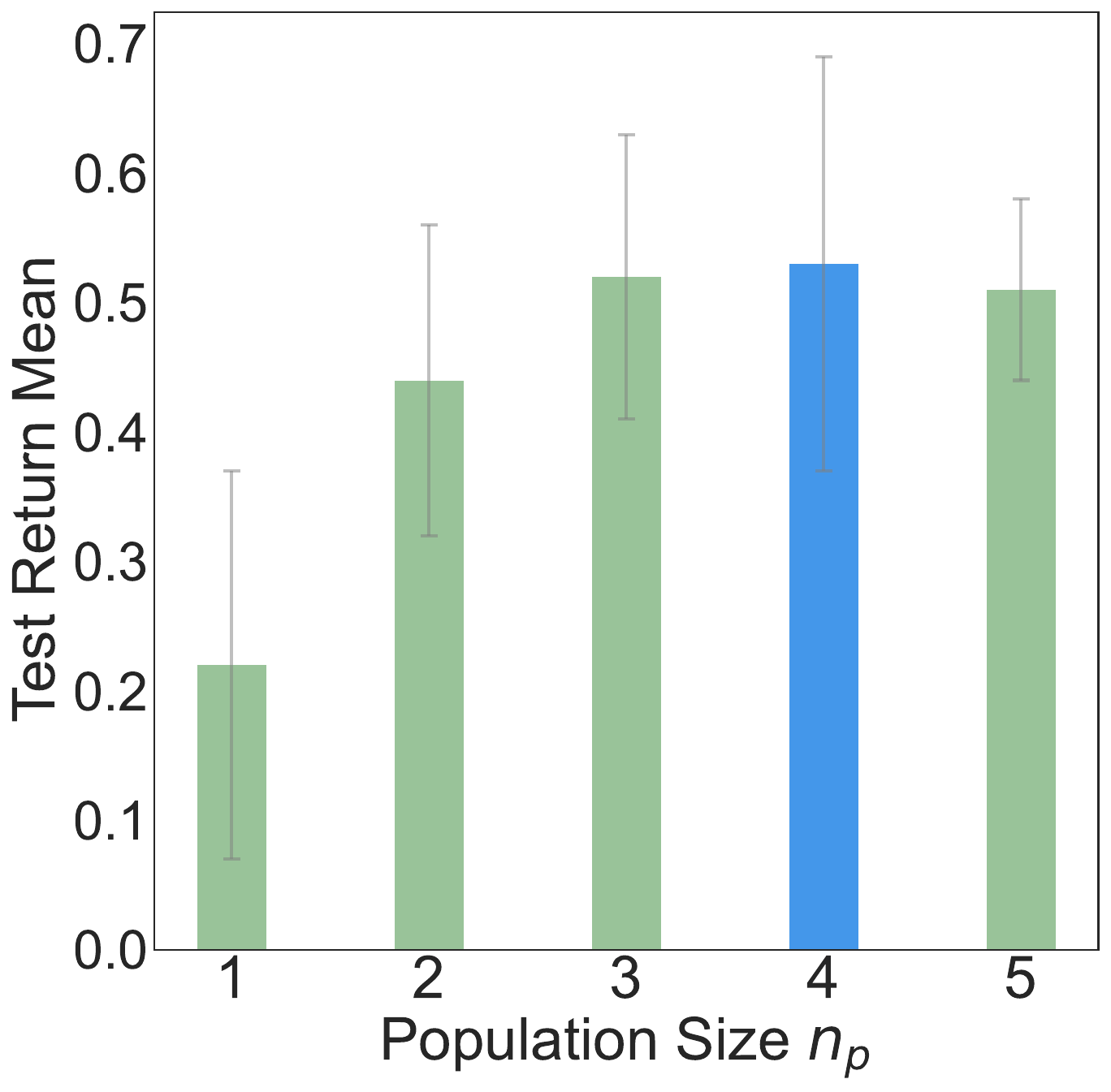}
    }
    \subfigure[]{
    \label{stop threshold}
    \includegraphics[width=0.26\textwidth]{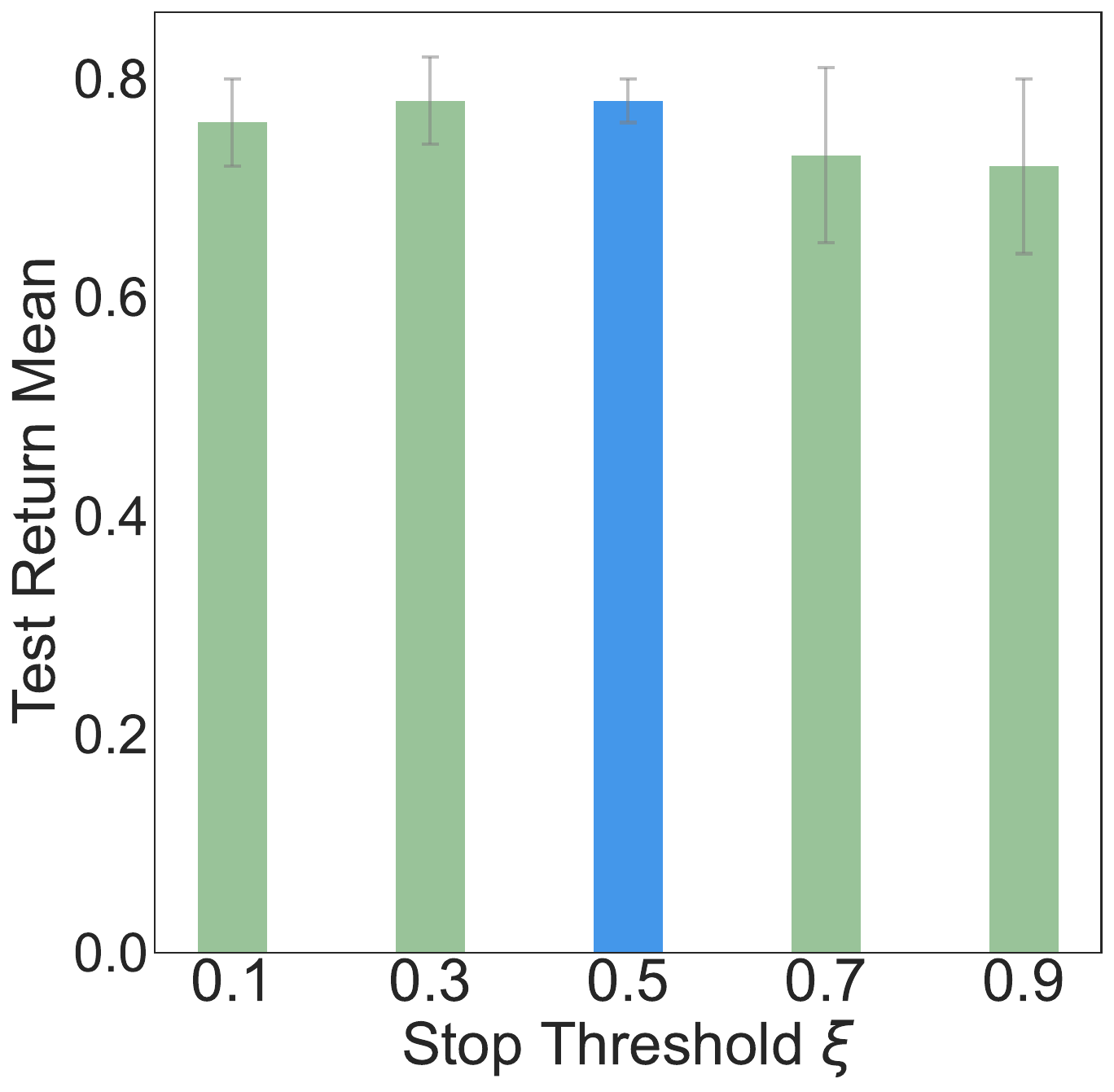}
    }
    \subfigure[]{
    \label{head threshold}
    \includegraphics[width=0.26\textwidth]{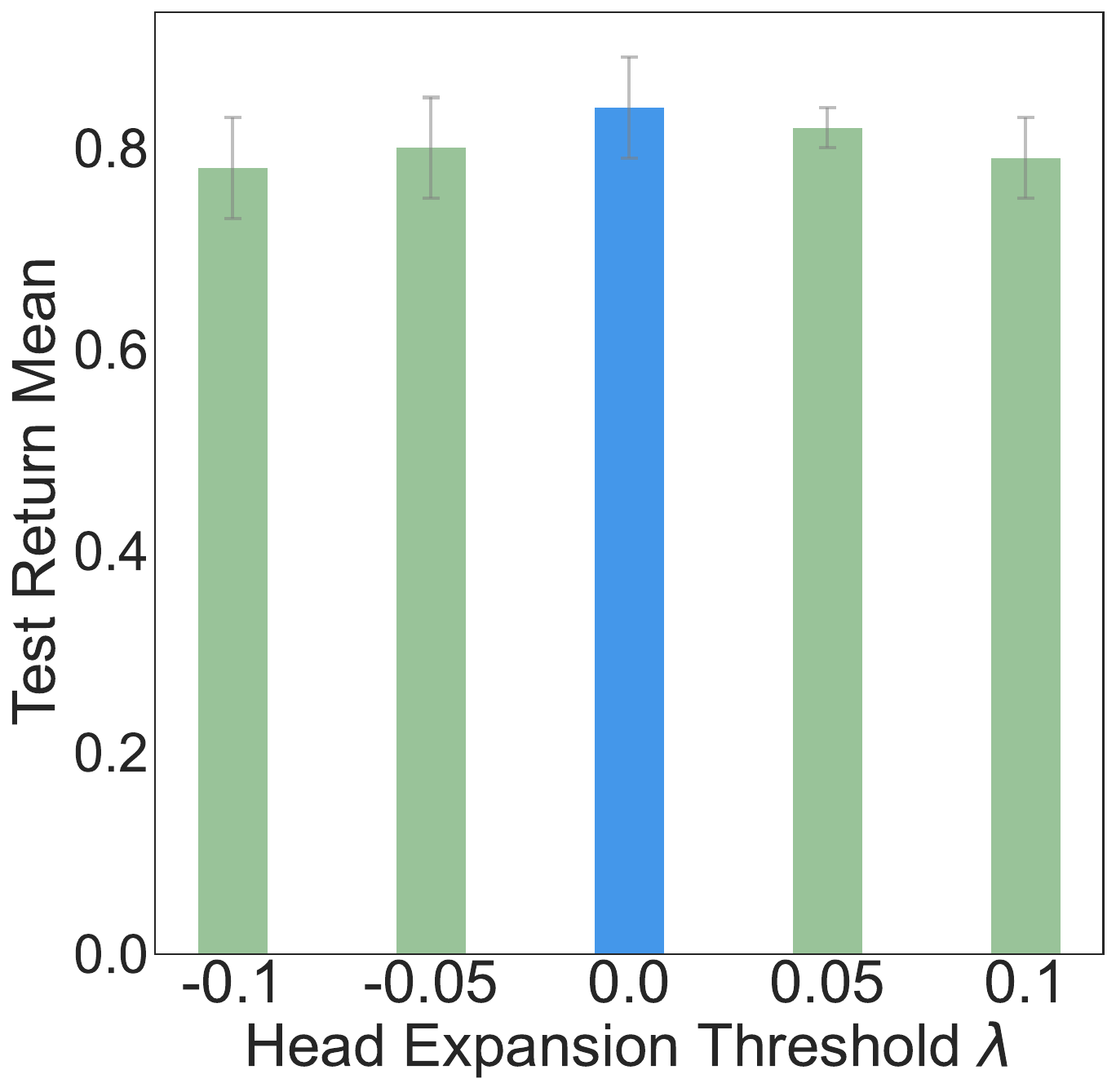}
    }
  \caption{The complete sensitivity analysis.}
  \label{sensitive appendix}
\end{figure}

\end{document}